\documentclass[reqno,10pt,a4paper]{amsart}

\usepackage{amssymb,mathptmx,eucal,mathrsfs,array,setspace,color,geometry,enumitem,cite}
\usepackage{graphicx}
\usepackage{tikz}
\usetikzlibrary{positioning}

\DeclareSymbolFont{largesymbols}{OMX}{zplm}{m}{n} %Replaces summation symbol in times by palatino one...

\setitemize{leftmargin=*}                % Removes margin from itemized lists
\setenumerate{leftmargin=*,              % Removes margin for enumerate lists       (enumitem package)
              label=\textup{(\arabic*)}} % makes labels (1) in lists and references

\newcolumntype{C}{>{$}c<{$}} %Defines math mode in tabular (array package)

\numberwithin{equation}{section}

\allowdisplaybreaks

\DeclareMathOperator{\im}{im}

\renewcommand{\ge}{\geqslant}
\renewcommand{\le}{\leqslant}

\newcommand{\alg}[1]{\mathfrak{#1}} % for algebras
\newcommand{\grp}[1]{\mathsf{#1}} % for groups

\newcommand{\func}[2]{#1 \left( #2 \right)} % standard variations are for display
\newcommand{\tfunc}[2]{#1 \bigl( #2 \bigr)} % t-variations are for using in text

\newcommand{\brac}[1]{\left( #1 \right)}

\newcommand{\sqbrac}[1]{\left[ #1 \right]}

\newcommand{\set}[1]{\left\{ #1 \right\}}

\newcommand{\st}{\mspace{5mu} : \mspace{5mu}} % "such that" in sets

\newcommand{\abs}[1]{\left\lvert #1 \right\rvert}

\newcommand{\ZZ}{\mathbb{Z}}

\newcommand{\RR}{\mathbb{R}}
\newcommand{\CC}{\mathbb{C}}

\newcommand{\pd}{\partial}     % holomorphic partial d
\newcommand{\ahol}[1]{\overline{#1}}

\newcommand{\dd}{\mathrm{d}}   % d in derivatives and integrals
\newcommand{\ii}{\mathfrak{i}} % imaginary unit
\newcommand{\ee}{\mathsf{e}}   % ln e = 1

\newcommand{\wun}{\mathbf{1}}  % the unit of all sorts of things

\newcommand{\normord}[1]{\mbox{${} : #1 : {}$}} % normal ordering ({} necessary to prevent := or =:)

\newcommand{\comm}[2]{\bigl[ #1 , #2 \bigr]}

\newcommand{\Vac}{\Omega}

\newcommand{\Ra}{\Rightarrow}
\newcommand{\ra}{\rightarrow}
\newcommand{\lra}{\longrightarrow}
\newcommand{\ses}[3]{0 \ra #1 \ra #2 \ra #3 \ra 0}                                  % short exact sequence
\newcommand{\dses}[5]{0 \lra #1 \overset{#2}{\lra} #3 \overset{#4}{\lra} #5 \lra 0} % displayed ses
\newcommand{\res}[4]{\cdots \lra #4 \lra #3 \lra #2 \lra #1 \lra 0}                 % resolution
\newcommand{\cores}[4]{0 \lra #1 \lra #2 \lra #3 \lra #4 \lra \cdots}               % coresolution
\DeclareMathOperator{\Ext}{Ext}
\newcommand{\Extgrp}[3]{\Ext^{#1} \bigl( #2 , #3 \bigr)}

\newcommand{\categ}[1]{\mathscr{#1}}   % requires mathrsfs package

\newcommand{\affine}[1]{\widehat{#1}}
\newcommand{\SLG}[2]{\grp{#1} \bigl( #2 \bigr)}                             % Lie groups like SU(2)
\newcommand{\SLA}[2]{\alg{#1} \bigl( #2 \bigr)}                             % Lie algebras like sl(2)
\newcommand{\AKMA}[2]{\affine{\alg{#1}} \left( #2 \right)}                  % Kac-Moody algebras
\newcommand{\AKMSA}[3]{\affine{\alg{#1}} \left( #2 \middle\vert #3 \right)} % Kac-Moody superalgebras

\newcommand{\Fin}[1]{\overline{#1}}
\newcommand{\Ghost}{\alg{G}}                                                % ghost algebra
\newcommand{\FinGhost}{\Fin{\Ghost}}                                        % horizontal algebra
\newcommand{\ExtGhost}[1]{\alg{E}_{#1}}                                     % simple current extension

\newcommand{\conjaut}{\mathsf{c}}                     % conjugation
\newcommand{\sfaut}{\sigma}                           % spectral flow
\newcommand{\conjmod}[1]{\tfunc{\conjaut}{#1}}        % conjugate an affine module
\newcommand{\sfmod}[2]{\tfunc{\sfaut^{#1}}{#2}}       % apply spectral flow #1 times to #2

\newcommand{\VacMod}{\mathcal{V}}          % vacuum module
\newcommand{\Typ}[1]{\mathcal{W}_{#1}}     % typical ghost module
\newcommand{\Stag}{\mathcal{P}}            % staggered ghost module

\newcommand{\ExtVacMod}{\mathbb{V}}          % extended vacuum module
\newcommand{\ExtTyp}[1]{\mathbb{W}_{#1}}     % typical extended module
\newcommand{\FinTyp}[1]{\Fin{\mathcal{W}}_{#1}}

\newcommand{\OthVac}{\mathsf{L}}
\newcommand{\OthTyp}[1]{\mathsf{E}_{#1}}
\newcommand{\OthStag}{\mathsf{S}}

\newcommand{\directint}{\ominus \mspace{-17.4mu} \int} % makes a symbol for a direct integral

\DeclareMathOperator{\tr}{tr}
\newcommand{\traceover}[1]{\tr_{\raisebox{-3pt}{$\scriptstyle #1$}}}
\newcommand{\chmap}{\mathrm{ch}}
\newcommand{\Gr}[1]{\bigl[ #1 \bigr]}           % element of a Grothendieck group/ring
\newcommand{\ch}[1]{\chmap \Gr{#1}}             % characters
\newcommand{\fch}[2]{\tfunc{\ch{#1}}{#2}}       % characters as functions of q and ...

\newcommand{\jth}[1]{\vartheta_{#1}}            % Jacobi theta
\newcommand{\fjth}[2]{\tfunc{\jth{#1}}{#2}}     % Jacobi theta as a function of z, q (or zeta, tau)

\newcommand{\modS}{\mathsf{S}}                        % modular S-matrix
\newcommand{\modT}{\mathsf{T}}                        % modular T-matrix
\newcommand{\modC}{\mathsf{C}}                        % modular conjugation matrix

\newcommand{\Smat}[2]{\modS \bigl[ #1 \ra #2 \bigr]}  % S-matrix entry

\newcommand{\fuse}{\mathbin{\times}}                              % fusion
\newcommand{\Grfuse}{\mathbin{\boxtimes}}                         % Verlinde algebra fusion
\newcommand{\fuscoeff}[3]{\genfrac{[}{]}{0pt}{0}{#3}{#1 \ \ #2}}  % fusion coefficient
\renewcommand{\coprod}[1]{\func{\Delta}{#1}}                        % fusion coproduct
\newcommand{\fcoprod}[2]{\func{\coprod{#1}}{#2}}

\newcommand{\othfuse}{\mathbin{\widehat{\times}}}

\newcommand{\eqnref}[1]{Equation~\eqref{#1}}
\newcommand{\eqnDref}[2]{Equations~\eqref{#1} and \eqref{#2}}
\newcommand{\secref}[1]{Section~\ref{#1}}

\newcommand{\appref}[1]{Appendix~\ref{#1}}

\newcommand{\figref}[1]{Figure~\ref{#1}}

\newcommand{\thmref}[1]{Theorem~\ref{#1}}

\newcommand{\propref}[1]{Proposition~\ref{#1}}

\newcommand{\lemref}[1]{Lemma~\ref{#1}}

\newcommand{\corref}[1]{Corollary~\ref{#1}}

\newcommand{\conjref}[1]{Conjecture~\ref{#1}}
\newcommand{\conjDref}[2]{Conjectures~\ref{#1} and \ref{#2}}

\newcommand{\cft}{conformal field theory}
\newcommand{\cfts}{conformal field theories}
\newcommand{\uea}{universal enveloping algebra}
\newcommand{\lcft}{logarithmic conformal field theory}
\newcommand{\lcfts}{logarithmic conformal field theories}
\newcommand{\WZW}{Wess-Zumino-Witten}
\newcommand{\PBW}{Poincar\'{e}-Birkhoff-Witt}
\newcommand{\NGK}{Nahm-Gaberdiel-Kausch}

\newcommand{\opes}{operator product expansions}
\newcommand{\hwv}{highest weight vector}
\newcommand{\hwvs}{highest weight vectors}
\newcommand{\lwv}{lowest weight vector}

\newcommand{\hwm}{highest weight module}
\newcommand{\hwms}{highest weight modules}
\newcommand{\voa}{vertex operator algebra}
\newcommand{\voas}{vertex operator algebras}

\newcommand{\lhs}{left-hand side}
\newcommand{\rhs}{right-hand side}

\theoremstyle{plain}
\newtheorem{thm}{Theorem}
\newtheorem{prop}[thm]{Proposition}
\newtheorem{lem}[thm]{Lemma}
\newtheorem{cor}[thm]{Corollary}
\theoremstyle{definition} % Non-italicised text

\newtheorem{conj}{Conjecture}

\begin{document}

\title{Bosonic Ghosts at $c=2$ as a Logarithmic CFT}

\author[D Ridout]{David Ridout}

\address[David Ridout]{
Department of Theoretical Physics \\
Research School of Physics and Engineering 
and
Mathematical Sciences Institute \\
Australian National University \\
Acton, ACT 2601 \\
Australia
}

\email{david.ridout@anu.edu.au}

\author[S Wood]{Simon Wood}

\address[Simon Wood]{
Department of Theoretical Physics \\
Research School of Physics and Engineering \\
Australian National University \\
Acton, ACT 2601 \\
Australia
}

\email{simon.wood@anu.edu.au}

\dedicatory{\textup{\today}}

\begin{abstract}
Motivated by Wakimoto free field realisations, the bosonic ghost system of central charge $c=2$ is studied using a recently proposed formalism for logarithmic conformal field theories.  This formalism addresses the modular properties of the theory with the aim being to determine the (Grothendieck) fusion coefficients from a variant of the Verlinde formula.  The key insight, in the case of bosonic ghosts, is to introduce a family of parabolic Verma modules which dominate the spectrum of the theory.  The results include S-transformation formulae for characters, non-negative integer Verlinde coefficients, and a family of modular invariant partition functions.  The logarithmic nature of the corresponding ghost theories is explicitly verified using the Nahm-Gaberdiel-Kausch fusion algorithm.

\vspace{2mm}

\end{abstract}

\maketitle

\onehalfspacing

\section{Introduction} \label{sec:Intro}

Ghost systems have a long history in \cft{}, particularly with regard to Faddeev-Popov gauge fixing of superstrings, see \cite{FriCon86} for example, but also as ingredients for constructing more complicated theories, Wakimoto realisations of \WZW{} models \cite{WakFoc86} and quantum hamiltonian reductions \cite{FeiAff92} being notable examples.  The intrinsic appeal of ghost systems is that they are examples of free field theories.  On the other hand, these theories are strongly non-unitary and, in the case of bosonic ghost systems, the spectrum of conformal weights is well known to be unbounded below.

The fermionic ghost system with central charge $c=-2$ has received much attention under the guise of the symplectic fermions \cft{} and is known to be logarithmic \cite{GurLog93,KauCur95,KauSym00}.  The logarithmic nature of the bosonic ghost system with $c=-1$ then follows easily from its realisation as a pair of symplectic fermions coupled to a lorentzian boson \cite{LesLog04}.  The fact that the different ghost systems may all be regarded as the same theory with a different choice of energy-momentum tensor now strongly suggests that every ghost system is a \lcft{}.  \thmref{thm:FusionTT} below confirms this conclusion for the $c=2$ bosonic ghost theory.

This confirmation is not, however, the aim of this note.  Rather, we wish to illustrate how the recently proposed standard module formalism \cite{CreLog13} for \lcfts{} allows one to efficiently analyse the \mbox{$c=2$} bosonic ghost system, in particular its modular properties, fusion rules\footnote{We will denote the fusion product of \voa{} modules by $\fuse$, reserving the symbol $\otimes$, and the term ``tensor product'', for the tensor product of complex vector spaces.  While fusion is expected to have the properties of an abstract tensor product, this has only been proven under certain assumptions on the algebra and the category of modules, see \cite{HuaTen13} for example, that do not seem to be met in this paper.} and modular invariant partition functions.  Because this formalism is tailored to studying the modular properties of the theory's characters, we could have chosen any bosonic system in which the ghost fields have integer conformal weight (to facilitate the T-transformation of characters).  The choice $c=2$ is convenient and it reflects our interest in Wakimoto free field realisations.  As mentioned above, we expect that the results are broadly independent of the choice of $c$.

Of course, the $c=-1$ bosonic ghost system is already very well understood.  However, the analysis in this case proceeds by considering the $\ZZ_2$-orbifold theory that coincides with the fractional level \WZW{} model $\AKMA{sl}{2 ; \RR}_{-1/2}$.  Being non-free, one has to work fairly hard to establish the modular transformation properties \cite{CreMod12} and fusion rules \cite{RidFus10} for this affine algebra.  Here, it is very important to realise that category $\categ{O}$ is not sufficient --- the physically relevant module category is far larger.  It is, unfortunately, not clear how to determine this physically relevant category.  We can only insist that it be closed under fusion and conjugation, as well as have the property that one can construct a modular invariant partition function from the characters.  Granting these results for $\AKMA{sl}{2 ; \RR}_{-1/2}$, the fusion rules may then be lifted to the $c=-1$ bosonic ghost system using the technology of simple current extensions \cite{SchSim90} and we summarise them, for ease of comparison, in \appref{app:c=-1Fusion}.\footnote{This technology was originally developed for rational \cfts{}, but also applies, with almost no changes, to non-rational theories to which the standard module formalism applies.  This formalism is refined, and the connection to simple currents outlined, in \cite{RidVer14}.}  The analysis reported here for $c=2$ is significantly more straightforward.  Indeed, this transparency leads us to propose that the $c=2$ bosonic ghost system should be regarded as an archetypal example of a \lcft{} \cite{CreRel11}.

We start by reviewing the general bosonic ghost system in \secref{sec:Algebra} to fix conventions and notation.  As usual, particular attention is paid to the conjugation and spectral flow automorphisms of the mode algebra.  \secref{sec:Reps} then addresses the representation theory, starting with \hwms{} (but for general Borel subalgebras sharing the chosen Cartan subalgebra).  We find that there is a unique \hwm{} for each Borel subalgebra.  The more interesting case of parabolic \hwms{} is then studied, anticipating their necessity for modular transformations.\footnote{The importance of the parabolic modules appears to have been largely overlooked in previous studies \cite{LesSU202,LesLog04}, leading to incorrect conclusions concerning modularity and the inapplicability of the Verlinde formula.  Here, we are guided by \cite{RidSL210,CreMod12} where the analogous modules are found to be crucial for a complete understanding of the $\AKMA{sl}{2;\RR}_{-1/2}$ model.}  For each Borel subalgebra, we choose a certain parabolic subalgebra and find a continuum of parabolic \hwms{} parametrised by $\RR / \ZZ$.  We refer to \cite{HumRep08} for definitions and basic properties of parabolic subalgebras and modules (in the context of semisimple Lie algebras).

Characters are then computed in \secref{sec:Char}, where we quickly detail the contradictions inherent in regarding them as meromorphic functions (on the product of the Riemann sphere and an open disc), see also \cite{RidSL208}.  We would like to strongly emphasise that writing bosonic ghost characters in terms of modular forms can therefore lead to incorrect conclusions.  Proceeding instead in a distributional setting \cite{CreMod12}, we determine expressions for the characters of the parabolic modules and construct resolutions to deduce formulae for the highest weight characters.  The core of the analysis now follows in \secref{sec:Mod}.  There, we apply the modular S-transformation to the parabolic characters in \thmref{thm:SMatrix} and show that the result defines a unitary integral operator on the space spanned by the characters that is symmetric with respect to the canonical basis and squares to the conjugation map.  This turns out to involve a surprisingly non-trivial automorphy factor which is dealt with by a judicious extension of the characters and their transformation properties.  S-transformation formulae for the highest weight characters follow.

We then apply the Verlinde formula in \secref{sec:Verlinde}, showing explicitly that the resulting ``Verlinde product rules'', for decomposing products of characters, have non-negative integer coefficients (\thmref{thm:GrFusion}).  Conjecturing that these rules coincide with the image of the fusion rules in the Grothendieck ring, we effortlessly arrive at almost all of the fusion rules involving simple modules.  The remaining simple-simple fusion rules are then calculated explicitly in \secref{sec:Fusion}.  This is a technical matter utilising the \NGK{} algorithm \cite{NahQua94,GabInd96}.  The actual computations are relatively simple, but there is a conceptual problem to overcome in that the modules that we would like to fuse all have trivial ``special subspaces''.  Nevertheless, a careful analysis shows (\thmref{thm:FusionTT}) that the resulting fusion products are staggered modules in the sense of \cite{RidSta09}, proving that the $c=2$ bosonic ghost system is a \lcft{}.  We use these results to propose a candidate for the physically relevant category of ghost modules (see \conjref{conj:Proj}).  This category is demonstrated to be closed under fusion and conjugation and we construct from it an infinite series of modular invariant partition functions, as required.  We close with a brief discussion of a bulk state space that we believe corresponds to the diagonal partition function.  The appendix recalls the fusion rules of the $c=-1$ bosonic ghost system, as derived in \cite{RidFus10}, for comparison.

\section{Ghost Algebras} \label{sec:Algebra}

The bosonic ghost system is generated by two (mutually bosonic) fields $\tfunc{\beta}{z}$ and $\tfunc{\gamma}{z}$, subject to the following \opes{}:
\begin{equation} \label{OPE:Ghosts}
\func{\beta}{z} \func{\beta}{w} \sim 0, \qquad 
\func{\beta}{z} \func{\gamma}{w} \sim -\frac{1}{z-w}, \qquad 
\func{\gamma}{z} \func{\gamma}{w} \sim 0.
\end{equation}
From these fields, one constructs a unique $\AKMA{gl}{1}$ current $\tfunc{J}{z}$ as follows:
\begin{equation} \label{eq:DefJ}
\func{J}{z} = \normord{\func{\beta}{z} \func{\gamma}{z}}.
\end{equation}
This current is lorentzian and it gives $\tfunc{\beta}{z}$ a charge of $+1$, whereas $\tfunc{\gamma}{z}$ is given charge $-1$.  In contrast, the conformal structure is not unique.  The bosonic ghost system admits a one-parameter family of energy-momentum tensors $\func{T^a}{z}$ parametrised by $a$:
\begin{equation} \label{eq:DefT}
\func{T^a}{z} = \brac{a-\tfrac{1}{2}} \normord{\func{\beta}{z} \func{\pd \gamma}{z}} + \brac{a+\tfrac{1}{2}} \normord{\func{\pd \beta}{z} \func{\gamma}{z}}.
\end{equation}
Although we may take $a \in \RR$ (or $\CC$), we will restrict $a$ to be in $\tfrac{1}{2} \ZZ$ for technical reasons to be discussed shortly.  We remark that $\func{J}{z}$ is only a conformal primary when $a=0$.  The central charge and conformal weights assigned to $\tfunc{\beta}{z}$ and $\tfunc{\gamma}{z}$ are then
\begin{equation} \label{eq:hc}
c^a = 12a^2 - 1 \in \ZZ; \qquad h_{\beta}^a = \tfrac{1}{2}-a \in \tfrac{1}{2} \ZZ, \quad h_{\gamma}^a = \tfrac{1}{2}+a \in \tfrac{1}{2} \ZZ.
\end{equation}
Note that $h_{\beta}^a + h_{\gamma}^a = 1$, in accordance with \eqref{OPE:Ghosts}.

Expanding the ghost fields (in the untwisted sector) as
\begin{equation}
\tfunc{\beta}{z} = \sum_{n \in \ZZ - h_{\beta}^a} \beta_n z^{-n-h_{\beta}^a}, \qquad 
\tfunc{\gamma}{z} = \sum_{n \in \ZZ - h_{\gamma}^a} \gamma_n z^{-n-h_{\gamma}^a},
\end{equation}
the commutation relations corresponding to \eqref{OPE:Ghosts} are
\begin{equation} \label{eq:Comm}
\comm{\beta_m}{\beta_n} = 0, \qquad 
\comm{\beta_m}{\gamma_n} = -\delta_{m+n=0} \: \wun, \qquad 
\comm{\gamma_m}{\gamma_n} = 0.
\end{equation}
We will denote by $\Ghost$ the infinite-dimensional complex Lie algebra spanned by the $\beta_n$, $\gamma_n$ and the central element $\wun$, equipped with the Lie brackets \eqref{eq:Comm}.  We also identify $\wun$ with the unit of the \uea{} of $\Ghost$ and assume that it acts as the identity operator on any $\Ghost$-module.  The subspace $\CC \wun$ will be referred to as the Cartan subalgebra of $\Ghost$.

The Lie algebra $\Ghost$ admits several useful automorphisms that preserve this Cartan subalgebra.  In particular, we mention the conjugation automorphism $\conjaut$ and the spectral flow automorphisms $\sfaut^{\ell}$ which act on the generators $\beta_n$ and $\gamma_n$ as follows:
\begin{equation}
\conjmod{\beta_n} = \gamma_n, \quad \conjmod{\gamma_n} = -\beta_n; \qquad 
\sfmod{\ell}{\beta_n} = \beta_{n-\ell}, \quad \sfmod{\ell}{\gamma_n} = \gamma_{n+\ell}.
\end{equation}
Note that $\conjaut \sfaut^{\ell} = \sfaut^{-\ell} \conjaut$.  We remark that conjugation does not have order $2$ as one might expect, but instead has order $4$.\footnote{As the action of $\conjaut^2$ may be identified with that of $-\wun$, conjugation still defines an order $2$ permutation on modules, as one would expect.}  We also note that $\conjaut$ preserves the mode index so that, for example, if $\beta_n$ has $n \in \ZZ - h_{\beta}^a$, then $\gamma_n = \conjmod{\beta_n}$ has $n \in \ZZ - h_{\beta}^a = \ZZ + h_{\gamma}^a$, rather than $n \in \ZZ - h_{\gamma}^a$.  It follows that unless $h_{\beta}^a$ and $h_{\gamma}^a$ belong to $\tfrac{1}{2} \ZZ$, conjugation will not preserve the untwisted sector.\footnote{Equivalently, the conjugate of an untwisted module will be twisted in general.  There is no inconsistency here, but it does simplify matters if we assume that $h_{\beta}^a, h_{\gamma}^a \in \tfrac{1}{2} \ZZ$.  More importantly, we know of no physical applications of ghost systems with conformal weights not in $\tfrac{1}{2} \ZZ$.}  This is why we are explicitly assuming that $a \in \tfrac{1}{2} \ZZ$.  Similarly, the spectral flow automorphism $\sfaut^{\ell}$ will only preserve the untwisted sector if $\ell \in \ZZ$.

These automorphisms may be used to construct families of $\Ghost$-modules by twisting the action on any given module.  So, let $\mathcal{M}$ be a $\Ghost$-module and define new $\Ghost$-modules $\conjmod{\mathcal{M}}$ and $\sfmod{\ell}{\mathcal{M}}$ as follows.  First, we define $\conjmod{\mathcal{M}}$ and $\sfmod{\ell}{\mathcal{M}}$ as vector spaces isomorphic to $\mathcal{M}$:
\begin{equation}
\conjmod{\mathcal{M}} = \set{\conjmod{v} \st v \in \mathcal{M} \vphantom{\sfmod{\ell}{v}}}, \qquad 
\sfmod{\ell}{\mathcal{M}} = \set{\sfmod{\ell}{v} \st v \in \mathcal{M}}.
\end{equation}
Here, the symbols $\conjmod{v}$ and $\sfmod{\ell}{v}$ are formal so that the isomorphisms are given by $v \mapsto \conjmod{v}$ and $v \mapsto \sfmod{\ell}{v}$, respectively.  These vector spaces are then equipped with the following $\Ghost$-action:
\begin{equation}
\alpha \cdot \conjmod{v} = \conjmod{\tfunc{\conjaut^{-1}}{\alpha} v}, \quad 
\alpha \cdot \sfmod{\ell}{v} = \sfmod{\ell}{\sfmod{-\ell}{\alpha} v}, \qquad 
\text{for all \(\alpha \in \Ghost\).}
\end{equation}
We will refer to the module $\conjmod{\mathcal{M}}$ as the conjugate of $\mathcal{M}$ and to the $\sfmod{\ell}{\mathcal{M}}$ as the spectral flow images of $\mathcal{M}$.

Of course, one can always identify the vector space underlying $\mathcal{M}$ with that of $\conjmod{\mathcal{M}}$ and the $\sfmod{\ell}{\mathcal{M}}$, instead of making the isomorphism explicit.  Then, one needs to distinguish the $\Ghost$-action, for example by making the representation explicit.  In particular, if $\rho$ denotes the representation of $\Ghost$ on the vector space $\mathcal{M}$, then the conjugate and spectrally flowed representations are defined by
\begin{equation}
\tfunc{\rho_{\conjaut}}{\alpha} = \tfunc{\rho}{\tfunc{\conjaut^{-1}}{\alpha}}, \qquad \tfunc{\rho_{\ell}}{\alpha} = \tfunc{\rho}{\sfmod{-\ell}{\alpha}}.
\end{equation}
As we prefer the language of modules over representations, we will keep the vector space isomorphisms explicit.  It is not hard to translate between the two languages if desired.

The induced action of the conjugation and spectral flow automorphisms on the current and Virasoro modes is most easily computed by lifting the automorphisms to the level of fields.  The results are:
\begin{equation}
\begin{aligned}
\conjmod{J_n} &= -J_n - 2a \delta_{n=0} \: \wun, \\
\conjmod{L_n^a} &= L_n^a + 2a n J_n,
\end{aligned}
\qquad
\begin{aligned}
\sfmod{\ell}{J_n} &= J_n + \ell \delta_{n=0} \: \wun, \\
\sfmod{\ell}{L_n^a} &= L_n^a - \ell J_n - \ell \brac{a + \tfrac{1}{2} \ell} \delta_{n=0} \: \wun.
\end{aligned}
\end{equation}
We note that the charge and conformal weight of a weight vector $v \in \mathcal{M}$ change as follows upon conjugating or applying spectral flow:
\begin{equation} \label{eq:FlowedWeights}
\begin{aligned}
J_0 v &= j v, \\
L_0^a v &= h v
\end{aligned}
\qquad \Ra \qquad 
\begin{aligned}
J_0 \conjmod{v} &= -\brac{j+2a} \conjmod{v}, \\
L_0^a \conjmod{v} &= h \conjmod{v},
\end{aligned}
\qquad 
\begin{aligned}
J_0 \sfmod{\ell}{v} &= \brac{j-\ell} \sfmod{\ell}{v}, \\
L_0^a \sfmod{\ell}{v} &= \sqbrac{h + \ell j + \ell \brac{a - \tfrac{1}{2} \ell}} \sfmod{\ell}{v}.
\end{aligned}
\end{equation}
The fact that the $\sigma^{\ell}$ do not preserve $L_0^a$, hence the conformal weights, is the origin of the name ``spectral flow''.

We close this section by remarking that the elements of the one-parameter family $\func{T^a}{z}$ may be viewed as deformations of the element with $a=0$.  Indeed,
\begin{equation} \label{eq:TDef}
\func{T^a}{z} = \func{T^0}{z} + a \func{\partial J}{z}.
\end{equation}
It follows that each bosonic ghost system shares the same representation content, independent of $a$.\footnote{In fact, the evidence at hand not only suggests that the modules over the ghost \voas{} form equivalent abelian categories, but that the equivalence extends to tensor categories.  In other words, the fusion rules of the ghost theories are also identical.  We hope to make this more precise in the future.}  We may therefore choose a convenient representative to study in detail.  As mentioned above, the theory with $a=0$ and $c^0=-1$ received a full treatment in \cite{RidSL208,RidFus10,CreMod12}; however, the results were derived as a consequence of those for the fractional level model $\AKMA{sl}{2 ; \RR}_{-1/2}$, essentially because $h_{\beta}^0 = h_{\gamma}^0 = \tfrac{1}{2}$ leads to twisted modules whose characters are not eigenvectors for the modular T-transformation.  In what follows, we will instead specialise to $a=-\tfrac{1}{2}$ and set $h_{\beta} = 1$, $h_{\gamma} = 0$ and $c=2$, dropping the label $a$ from all quantities for simplicity.  This choice facilitates a direct investigation of the spectrum, modular properties of the characters, and fusion rules.  It also has the advantage of being of significant mathematical and physical interest through the Wakimoto free field realisations of affine Kac-Moody algebras \cite{WakFoc86}.

\noindent 

\section{Representation Theory} \label{sec:Reps}

As mentioned above, we will now choose the conformal structure, once and for all, so that $a=-\tfrac{1}{2}$, $h_{\beta} = 1$, $h_{\gamma} = 0$ and $c=2$.  The ghost algebra $\Ghost$ is not a (generalised) Kac-Moody algebra, but it does admit triangular decompositions with Cartan subalgebra $\CC \wun$.  In particular, we introduce a family of triangular decompositions, parametrised by $\ell \in \ZZ$, wherein the positive subalgebra is spanned by the $\beta_{n-\ell}$ and the $\gamma_{n+\ell+1}$, with $n \ge 0$.  These decompositions are clearly mapped into one another by the spectral flow automorphisms (and conjugation).  We will refer to the triangular decomposition with $\ell=0$ as the \emph{normal} decomposition.

Given a triangular decomposition, we may construct Verma modules.  As the Cartan subalgebra is spanned by the central element $\wun$, which we assume always acts as the identity operator, there is a unique Verma module for each decomposition.  We shall denote this Verma module by $\VacMod$ in the case that the decomposition is the normal one.  This module is generated by a state $\Vac$ which is annihilated by the $\beta_n$ and $\gamma_{n+1}$, for $n \ge 0$, hence it is annihilated by $J_0$, $L_0$ and $L_{-1}$.  We may therefore take $\Vac$ to be the (translation-invariant) vacuum of the bosonic ghost theory.  The vacuum Verma module $\VacMod$ is simple because any vector annihilated by the positive modes is also annihilated by $L_0$, so has conformal weight $0$, and the vectors $\gamma_0^n \Vac \in \VacMod$ of conformal weight $0$ are easily checked to be cyclic.  As is well known, this vacuum module admits the structure of a \voa{}.

The Verma modules obtained from the other triangular decompositions are then precisely the spectral flow images of the vacuum Verma module $\VacMod$.  We remark that the vector $\omega = \conjmod{\Vac} \in \conjmod{\VacMod}$ has charge $1$ and conformal weight $0$, by \eqref{eq:FlowedWeights}, so the module it generates is not isomorphic to $\VacMod$.  The ghost vacuum module is therefore \emph{not} self-conjugate.  Indeed, it is easy to check that the $\sfmod{\ell}{\VacMod}$ are all mutually non-isomorphic and that $\conjmod{\VacMod} \cong \sfmod{-1}{\VacMod}$.  As $\sfaut$ is an automorphism of $\Ghost$, all the Verma modules $\sfmod{\ell}{\VacMod}$ are simple.

In categorical terms, the vacuum module $\VacMod$ is the only simple object in the analogue of category $\categ{O}$.  There are also the twisted versions $\categ{O}^{\ell}$ whose unique simple object is $\sfmod{\ell}{\VacMod}$.  As $\sfaut$ is an automorphism of $\Ghost$, spectral flow defines exact functors between the $\categ{O}^{\ell}$.  Each of these categories is semisimple because $\VacMod$ admits no non-split self-extension on which the Cartan element $\wun$ acts as the identity operator.  For in such an extension, $\ses{\VacMod}{\mathcal{W}}{\VacMod}$, any $\Vac' \in \mathcal{W}$ projecting onto the \hwv{} $\Vac$ of the quotient $\VacMod$ would be cyclic, so there would exist $U$ in the \uea{} of $\Ghost$ for which $U \Vac' = \Vac$, the \hwv{} of the submodule $\VacMod$.  As $U \Vac = 0$ and $\VacMod$ is a Verma module, $U$ is a sum of terms of the form $U' \beta_n$ or $U'' \gamma_{n+1}$, where $n \ge 0$.  But, $\beta_n \Vac' = 0$ and $\gamma_{n+1} \Vac' = 0$, for all such $n$, because $\VacMod$ has no non-zero vectors of $(J_0, L_0)$-eigenvalues $(1,-n)$ and $(-1,-n-1)$, respectively.

However, the representation theory of the ghost \voa{} is not limited to Verma modules and twisted versions of category $\categ{O}$.  One can also consider parabolic Verma modules; indeed, we shall see in \secref{sec:Mod} that we must.  Recall that a subalgebra of a Lie algebra with triangular decomposition $\alg{g} = \alg{g}_- \oplus \alg{g}_0 \oplus \alg{g}_+$ is said to be parabolic if it contains the Borel subalgebra $\alg{g}_0 \oplus \alg{g}_+$.  Given the normal triangular decomposition, say, there turn out to be infinitely many parabolic subalgebras because we may extend the normal Borel subalgebra by any combination of the negative $\beta_n$ modes and the non-positive $\gamma_n$ modes.  Parabolic subalgebras containing the other Borel subalgebras may then be obtained through spectral flow (and conjugation).

This plethora of parabolic subalgebras turns out to be surplus to our needs.  For the analysis to follow, we will only require one of the parabolic subalgebras extending the normal Borel subalgebra, as well as its cousins obtained by applying spectral flow.  The reason for ignoring the remaining parabolic subalgebras, and their associated parabolic Verma modules, will not be detailed here.  Suffice to say, the point is that we want these structures to be compatible with the entire mode algebra of the ghost \voa{}, not just $\Ghost$.  The normal parabolic subalgebra that we require corresponds to extending the normal Borel subalgebra by $\beta_0$.  It is therefore spanned by the $\beta_n$ and $\gamma_n$, with $n \ge 0$, and $\wun$; we will denote it by $\alg{p}$.  For this choice, the parabolic Verma modules (also known as generalised Verma modules) are obtained by inducing any module over the subalgebra $\FinGhost$ spanned by $\beta_0$, $\gamma_0$ and $\wun$, this module being lifted to a module over $\alg{p}$ by letting the modes with positive index act as zero.

We therefore study the $\FinGhost$-modules that are the direct sum of their eigenspaces under $\Fin{J}_0 = \gamma_0 \beta_0$, these being the obvious candidates for weight modules over $\FinGhost$:
\begin{prop} \label{prop:FinGhReps}
\textcolor{white}{x}
\begin{enumerate}
\item The only highest weight $\FinGhost$-module is $\Fin{\VacMod} = \CC[\gamma_0] \Fin{\Vac}$, generated by a \hwv{} $\Fin{\Vac}$ satisfying $\beta_0 \Fin{\Vac} = 0$, hence $\Fin{J}_0 \Fin{\Vac} = 0$.  This module is simple. 
\item The only lowest weight $\FinGhost$-module is $\Fin{\conjmod{\VacMod}} = \CC[\beta_0] \Fin{\omega}$, generated by a \lwv{} $\Fin{\omega}$ satisfying $\gamma_0 \Fin{\omega} = 0$, hence $\Fin{J}_0 \Fin{\omega} = \Fin{\omega}$.  This module is simple.
\item \label{it:ParaClass} There is, in addition, a continuous family of $\FinGhost$-modules parametrised by $[\lambda] \in \CC / \ZZ$.  They have a basis consisting of vectors $\Fin{u}_j$, with $j \in [\lambda] = \ZZ + \lambda$, satisfying $\Fin{J}_0 \Fin{u}_j = j \Fin{u}_j$.  
\begin{enumerate}[label=\textup{(\alph*)}]
\item When $[\lambda] \neq [0]$, these modules are simple and are denoted by $\FinTyp{\lambda}$.  We may realise $\FinTyp{\lambda}$ on $\CC[\beta_0] \Fin{u}_{\lambda} \oplus \CC[\gamma_0] \gamma_0 \Fin{u}_{\lambda}$, noting that $\FinTyp{\lambda} = \FinTyp{\mu}$ when $\lambda - \mu \in \ZZ$. 
\item When $[\lambda] = [0]$, there are two inequivalent indecomposable modules, $\FinTyp{0}^+$ and $\FinTyp{0}^-$, whose isomorphism classes are determined by the following short exact sequences ($\Extgrp{1}{\Fin{\conjmod{\VacMod}}}{\Fin{\VacMod}} = \Extgrp{1}{\Fin{\VacMod}}{\Fin{\conjmod{\VacMod}}} = \CC$):
\begin{equation}
\dses{\Fin{\VacMod}}{}{\FinTyp{0}^+}{}{\Fin{\conjmod{\VacMod}}}, \qquad 
\dses{\Fin{\conjmod{\VacMod}}}{}{\FinTyp{0}^-}{}{\Fin{\VacMod}}.
\end{equation}
Both may be realised on the space $\CC[\gamma_0] \Fin{u}_0 \oplus \CC[\beta_0] \Fin{u}_1$, where $\beta_0 \Fin{u}_0 = 0$ and $\gamma_0 \Fin{u}_1 = a^+ \Fin{u}_0$, for $\FinTyp{0}^+$, and $\beta_0 \Fin{u}_0 = a^- \Fin{u}_1$ and $\gamma_0 \Fin{u}_1 = 0$, for $\FinTyp{0}^-$.  We may normalise the basis vectors so that $a^+ = a^- = 1$.
\end{enumerate}
\end{enumerate}
\end{prop}
\noindent This classification is well known because $\FinGhost$ is the Weyl algebra $A_1$, also known as the canonical commutation relations algebra.  Indeed, Block classified \emph{all} simple modules over $A_1$ in \cite{BloIrr79}.  However, the proof for simple weight modules is quite easy, see \cite[Sec.~3.4]{MazLec10} for a similar proof for $\SLA{sl}{2}$, so we present a sketch for completeness.
\begin{proof}[Proof (sketch).]
As the Cartan subalgebra is spanned by $\wun$, which always acts as the identity, there is a unique Verma module and it is easy to verify that it is simple.  This takes care of (1).  (2) now follows by applying conjugation.

For (3), we need to know that a simple weight $\FinGhost$-module has one-dimensional weight spaces.  This follows by considering each weight space as a module over $\CC[\Fin{J}_0]$ and showing that these modules are simple.  The argument is by contradiction and uses only the \PBW{} theorem:  If a weight space has a proper non-zero $\CC[\Fin{J}_0]$-submodule, then it generates a proper non-zero $\FinGhost$-module.  (3a) now follows because we may normalise the weight vectors $\Fin{u}_j \in \FinTyp{\lambda}$ so that $\gamma_0 \Fin{u}_j = \Fin{u}_{j-1}$ and then, $\Fin{J}_0 \Fin{u}_j = j \Fin{u}_j$ implies that $\beta_0 \Fin{u}_j = j \Fin{u}_{j+1}$.  The existence of the $\FinTyp{\lambda}$ follows from their explicit construction.  (3b) likewise follows, with the extension groups being essentially parametrised by the coefficients $a^{\pm}$.
\end{proof}

\noindent We remark that because we do not seem to need complex weights in physical theories, we will throughout restrict the parameter $[\lambda]$ appearing in item \ref{it:ParaClass} above (and elsewhere) to lie in $\RR / \ZZ$.

Inducing $\Fin{\VacMod}$ and $\Fin{\conjmod{\VacMod}}$ recovers the usual Verma modules $\VacMod$ and $\conjmod{\VacMod}$, respectively, over $\Ghost$.  However, inducing the $\FinTyp{\lambda}$ and $\FinTyp{0}^{\pm}$ results in new parabolic Verma modules that we shall denote by $\Typ{\lambda}$ and $\Typ{0}^{\pm}$, respectively.  These may also be regarded as examples of relaxed \hwms{} in the spirit of \cite{FeiEqu98}.  It follows from \propref{prop:FinGhReps} that these new modules are simple for $\lambda \notin \ZZ$ and are otherwise characterised by the exact sequences
\begin{equation} \label{es:Atyp}
\dses{\VacMod}{}{\Typ{0}^+}{}{\conjmod{\VacMod}}, \qquad 
\dses{\conjmod{\VacMod}}{}{\Typ{0}^-}{}{\VacMod}.
\end{equation}
We also have $\Typ{\lambda} = \Typ{\mu}$ whenever $\lambda - \mu \in \ZZ$.  Twisting by spectral flow now realises the parabolic Verma modules, $\sfmod{\ell}{\Typ{\lambda}}$, $\sfmod{\ell}{\Typ{0}^+}$ and $\sfmod{\ell}{\Typ{0}^-}$, that correspond to other parabolic subalgebras of $\Ghost$.  These parabolic Verma modules are all mutually non-isomorphic.

The category $\categ{O}$ is therefore a full subcategory of the category $\categ{P}$ of parabolic \hwms{} corresponding to the parabolic subalgebra $\alg{p}$.  An analogous statement holds for the categories obtained by twisting by $\sfaut^{\ell}$.  Note that $\categ{P}$ has an uncountable family $\Typ{\lambda}$, $[\lambda] \in \CC / \ZZ$, $[\lambda] \neq [0]$, of inequivalent simple objects, as well as $\VacMod$ and $\conjmod{\VacMod}$.  This category is not semisimple because of \eqref{es:Atyp}, but the only non-semisimple block corresponds to $[\lambda] = [0]$.  However, we shall see in \secref{sec:Mod} that the physically relevant category must include not only $\categ{P}$, but also each of its spectrally-flowed versions, in order that the ghost characters span a representation of the modular group $\SLG{SL}{2;\ZZ}$.  We will also see in \secref{sec:Fusion} that closure under fusion leads to extensions between parabolic modules with different spectral flow indices.

To summarise (without categories), and to make contact with the standard module formalism of \cite{CreLog13,RidVer14}, we have constructed a continuous family of simple $\Ghost$-modules $\sfmod{\ell}{\Typ{\lambda}}$, parametrised by $[\lambda] \in \RR / \ZZ$, $[\lambda] \neq [0]$, and \mbox{$\ell \in \ZZ$}.  These parabolic Verma modules are the \emph{typical} modules.  The module conjugate to $\sfmod{\ell}{\Typ{\lambda}}$ is $\conjmod{\sfmod{\ell}{\Typ{\lambda}}} = \sfmod{-\ell}{\Typ{-\lambda}}$.  There are, moreover, two discrete families of indecomposable, but reducible, $\Ghost$-modules, $\sfmod{\ell}{\Typ{0}^+}$ and $\sfmod{\ell}{\Typ{0}^-}$, with simple composition factors $\sfmod{\ell}{\VacMod}$ and $\sfmod{\ell}{\conjmod{\VacMod}}$.  These modules are all \emph{atypical} and are also related by conjugation:  $\conjmod{\sfmod{\ell}{\Typ{0}^+}} = \sfmod{-\ell}{\Typ{0}^-}$ and $\conjmod{\sfmod{\ell}{\VacMod}} = \sfmod{-\ell-1}{\VacMod}$.  As the vacuum module $\VacMod$ is atypical, we expect that ghost theories will all be logarithmic.  The \emph{standard} modules of the theory are the typicals $\sfmod{\ell}{\Typ{\lambda}}$ and the indecomposable atypicals $\sfmod{\ell}{\Typ{0}^+}$ and $\sfmod{\ell}{\Typ{0}^-}$.  As we shall see, there is a uniform character formula for the standard modules and the corresponding modular S-transformations are straightforward to determine.

\section{Characters} \label{sec:Char}

Being a Verma module, the character of the vacuum module $\VacMod$ is easily found:
\begin{equation} \label{ch:Vac'}
\fch{\VacMod}{z;q} = \traceover{\VacMod} z^{J_0} q^{L_0 - c/24} \overset{!}{=} \frac{q^{-1/12}}{\prod_{i=1}^{\infty} \brac{1-zq^i} \brac{1-z^{-1} q^{i-1}}} = -\ii z^{1/2} \frac{\func{\eta}{q}}{\fjth{1}{z;q}}.
\end{equation}
Here, the ``$\overset{!}{=}$'' indicates that we are (temporarily) ignoring convergence regions by identifying the characters, which are formal power series, with their meromorphic continuations to $z \in \CC \cup \set{\infty}$ and $\abs{q} < 1$.  The character of the conjugate module $\conjmod{\VacMod}$ is similarly determined to be
\begin{equation} \label{ch:A1'}
\fch{\conjmod{\VacMod}}{z;q} \overset{!}{=} \frac{zq^{-1/12}}{\prod_{i=1}^{\infty} \brac{1-z^{-1}q^i} \brac{1-zq^{i-1}}} = -\ii z^{1/2} \frac{\func{\eta}{q}}{\fjth{1}{z^{-1};q}} = +\ii z^{1/2} \frac{\func{\eta}{q}}{\fjth{1}{z;q}}.
\end{equation}
It is not hard to check that these formulae are consistent with the identification $\conjmod{\VacMod} \cong \sfmod{-1}{\VacMod}$ using the properties of Jacobi theta functions and the relations
\begin{equation} \label{eq:TransformChars}
\fch{\conjmod{\mathcal{M}}}{z;q} = z \: \fch{\mathcal{M}}{z^{-1};q}, \qquad 
\fch{\sfmod{\ell}{\mathcal{M}}}{z;q} = z^{-\ell} q^{-\ell \brac{\ell+1} / 2} \: \fch{\mathcal{M}}{zq^{\ell};q},
\end{equation}
valid for any $\Ghost$-module $\mathcal{M}$.  However, they do lead to the suspicious identity of meromorphically-continued characters
\begin{equation} \label{eq:CharsCancel}
\ch{\VacMod} + \ch{\conjmod{\VacMod}} \overset{!}{=} 0
\end{equation}
which, when combined with the exact sequences \eqref{es:Atyp}, seems to say that the characters of the indecomposables $\Typ{0}^+$ and $\Typ{0}^-$ must vanish identically.

This erroneous conclusion is corrected \cite{RidSL208} by considering the difference between regarding characters as formal power series and regarding them as meromorphic functions.  The Dedekind eta and Jacobi theta functions converge for $\abs{q} < 1$, but the character formula \eqref{ch:Vac'} has poles whenever $z=q^i$, for some $i \in \ZZ$.  Thus, the character as a formal power series will only converge, upon interpreting $z$ and $q$ as complex numbers, to the given meromorphic function on one of the annuli in which the magnitude of $z$ is bounded between the magnitudes of two consecutive poles.  Indeed, the region of convergence of the vacuum character \eqref{ch:Vac'} is
\begin{equation}
\abs{q} < 1, \qquad 1 < \abs{z} < \abs{q}^{-1}.
\end{equation}
In general, the character of $\sfmod{\ell}{\VacMod}$ is only convergent in the region
\begin{equation}
\abs{q} < 1, \qquad \abs{q}^{-\ell} < \abs{z} < \abs{q}^{-\ell-1}.
\end{equation}
The regions of convergence of $\ch{\VacMod}$ and $\ch{\conjmod{\VacMod}} = \ch{\sfmod{-1}{\VacMod}}$ are therefore disjoint, so that while \eqref{eq:CharsCancel} may hold at the level of meromorphic functions, it makes no sense at the level of the characters (which are formal power series) themselves.  We therefore conclude that it is incorrect to treat characters as meromorphic functions in this case.\footnote{We also mention that it does not seem possible to instead consider characters as meromorphic functions with a given region of convergence.  One conceptual objection to this is that the modular S-transformation does not respect these convergence regions in any way, so it is not clear that characters with convergence regions may be subjected to modular analysis.}

Instead, we shall treat these formal power series as distributions over Laurent polynomials in $q$ and $z$.  This is suggested by the character formula for the typical modules $\Typ{\lambda}$ which obviously diverges everywhere if one tries to interpret it as a meromorphic function:
\begin{equation} \label{ch:T}
\ch{\Typ{\lambda}} = \sum_{n \in \ZZ} z^{n+\lambda} \frac{q^{-1/12}}{\prod_{i=1}^{\infty} \brac{1-zq^i} \brac{1-z^{-1}q^i}} = \sum_{n \in \ZZ} z^{n+\lambda} \frac{q^{-1/12}}{\prod_{i=1}^{\infty} \brac{1-q^i}^2} = \frac{z^{\lambda}}{\func{\eta}{q}^2} \sum_{n \in \ZZ} z^n.
\end{equation}
Here, we remark that the denominators in expressions such as these should be regarded as shorthand notation for the corresponding (geometric) power series.  This formula follows from the fact that a basis for the parabolic Verma module $\Typ{\lambda}$ may be chosen to consist of the parabolic \hwvs{} $u_j$, $j \in \ZZ + \lambda$, being acted upon freely by the negative modes $\beta_n$ and $\gamma_n$, $n<0$.  We have also noted that
\begin{equation}
\sum_{n \in \ZZ} \frac{z^n}{1-zq^i} = \sum_{n \in \ZZ} \sum_{k=0}^{\infty} z^{n+k} q^{ik} = \sum_{m \in \ZZ} \sum_{k=0}^{\infty} z^m q^{ik} = \sum_{m \in \ZZ} \frac{z^m}{1-q^i}.
\end{equation}

As an identity of formal power series (distributions), \eqref{ch:T} also holds for the atypical standards $\Typ{0}^+$ and $\Typ{0}^-$ upon substituting $\lambda = 0$.\footnote{We will often drop the label ``$\pm$'' when considering the characters of the atypical standard modules.}  Setting $z = \ee^{2 \pi \ii \zeta}$ now results in the divergent sum in \eqref{ch:T} being recognised as a singular distribution supported at $\zeta \in \ZZ$, that is $z=1$:
\begin{equation}
\sum_{n \in \ZZ} z^n = \sum_{n \in \ZZ} \ee^{2 \pi \ii n \zeta} = \sum_{m \in \ZZ} \func{\delta}{\zeta = m}.
\end{equation}
\eqnref{eq:CharsCancel} is therefore replaced, in this distributional setting, by
\begin{equation}
\ch{\VacMod} + \ch{\conjmod{\VacMod}} = \ch{\Typ{0}} = \frac{\sum_{m \in \ZZ} \func{\delta}{\zeta = m}}{\func{\eta}{q}^2},
\end{equation}
demonstrating that the \rhs{} is not $0$, but is rather a singular distribution supported at $z=1$.  We remark that $z=1$ is precisely the pole that separates the annuli of convergence of the characters on the \lhs{}.

Applying spectral flow then gives, using \eqnref{eq:TransformChars}, the characters of all the standard modules as distributions.  It therefore remains to compute the character of the vacuum module $\VacMod$, and its spectral flow images, as distributions rather than as meromorphic functions.  This is achieved by splicing the exact sequences \eqref{es:Atyp} with their spectrally-flowed counterparts to obtain resolutions
\begin{subequations}
\begin{equation}
\begin{aligned}
\res{\VacMod}{\sfmod{}{\Typ{0}^+}}{\sfmod{2}{\Typ{0}^+}}{\sfmod{3}{\Typ{0}^+}}&, \\
\res{\VacMod}{\Typ{0}^-}{\sfmod{-1}{\Typ{0}^-}}{\sfmod{-2}{\Typ{0}^-}} &
\end{aligned}
\end{equation}
or coresolutions
\begin{equation}
\begin{aligned}
&\cores{\VacMod}{\sfmod{}{\Typ{0}^-}}{\sfmod{2}{\Typ{0}^-}}{\sfmod{3}{\Typ{0}^-}}, \\
&\cores{\VacMod}{\Typ{0}^+}{\sfmod{-1}{\Typ{0}^+}}{\sfmod{-2}{\Typ{0}^+}}.
\end{aligned}
\end{equation}
\end{subequations}
We thereby deduce two character formulae for the vacuum module as a formal power series (distributions):
\begin{equation} \label{ch:Vac}
\ch{\VacMod} = \sum_{\ell = 1}^{\infty} \brac{-1}^{\ell-1} \ch{\sfmod{\ell}{\Typ{0}}}, \qquad 
\ch{\VacMod} = \sum_{\ell = 0}^{\infty} \brac{-1}^{\ell} \ch{\sfmod{-\ell}{\Typ{0}}}.
\end{equation}
The convergence of these expressions is meant in the following sense:  For each weight $(j,h)$, only a finite number of terms in either sum contribute to the multiplicity of $z^j q^h$.  We shall not dwell on the implication that the difference of these two expressions, a bi-infinite alternating sum of the atypical standard characters, vanishes.  Suffice to say that we regard either of these formulae as deciding on an appropriate topological completion of the span of the standard characters.  It is straightforward to check that the results which follow will not depend on which formula, hence which completion, we choose.

\section{Modular Transformations} \label{sec:Mod}

We prepare for computing S-transformations by calculating the character of a general standard module using \eqnDref{eq:TransformChars}{ch:T}:
\begin{equation}
\fch{\sfmod{\ell}{\Typ{\lambda}}}{z;q} = z^{-\ell} q^{-\ell \brac{\ell+1} / 2} \frac{z^{\lambda} q^{\ell \lambda}}{\func{\eta}{q}^2} \sum_{n \in \ZZ} z^n q^{n \ell} = \frac{z^{\lambda} q^{\ell \lambda + \ell \brac{\ell-1} / 2}}{\func{\eta}{q}^2} \sum_{n \in \ZZ} z^n q^{n \ell}.
\end{equation}
Writing $q = \ee^{2 \pi \ii \tau}$ and $z = \ee^{2 \pi \ii \zeta}$, this simplifies to
\begin{equation} \label{ch:Standards}
\fch{\sfmod{\ell}{\Typ{\lambda}}}{\zeta \big\vert \tau} = \frac{\ee^{\ii \pi \ell \brac{\ell-1} \tau}}{\func{\eta}{\tau}^2} \sum_{n \in \ZZ} \ee^{2 \pi \ii n \lambda} \func{\delta}{\zeta + \ell \tau = n}.
\end{equation}
\begin{thm} \label{thm:SMatrix}
The standard characters \eqref{ch:Standards} have S-transformation
\begin{subequations} \label{eq:SMatrix}
\begin{equation}
\fch{\sfmod{\ell}{\Typ{\lambda}}}{\zeta / \tau \big\vert {-1} / \tau} = \func{A}{\zeta \big\vert \tau} \sum_{m \in \ZZ} \int_{\RR / \ZZ} \Smat{\sfmod{\ell}{\Typ{\lambda}}}{\sfmod{m}{\Typ{\mu}}} \fch{\sfmod{m}{\Typ{\mu}}}{\zeta \big\vert \tau} \: \dd \mu,
\end{equation}
where
\begin{equation} \label{eq:TypSKernel}
\func{A}{\zeta \big\vert \tau} = \frac{\abs{\tau}}{-\ii \tau} \ee^{-\ii \pi \zeta^2 / \tau} \ee^{\ii \pi \zeta / \tau} \ee^{-\ii \pi \zeta}, \qquad 
\Smat{\sfmod{\ell}{\Typ{\lambda}}}{\sfmod{m}{\Typ{\mu}}} = \brac{-1}^{\ell+m} \ee^{-2 \pi \ii \brac{\ell \mu + m \lambda}}.
\end{equation}
\end{subequations}
\end{thm}
\noindent This theorem may be verified by direct substitution.  We omit the details.

Recall from \eqref{ch:Vac} that all characters may be expressed as (infinite) linear combinations of the standard characters \eqref{ch:Standards}.  The latter therefore form a (topological) basis for the space of characters.  In this basis, which we call the standard basis, the S-transformation is manifestly symmetric and unitary:
\begin{subequations}
\begin{gather}
\Smat{\sfmod{\ell}{\Typ{\lambda}}}{\sfmod{m}{\Typ{\mu}}} = \Smat{\sfmod{m}{\Typ{\mu}}}{\sfmod{\ell}{\Typ{\lambda}}}, \label{eq:SmatSymm} \\
\sum_{m \in \ZZ} \int_{\RR / \ZZ} \Smat{\sfmod{\ell}{\Typ{\lambda}}}{\sfmod{m}{\Typ{\mu}}} \Smat{\sfmod{n}{\Typ{\nu}}}{\sfmod{m}{\Typ{\mu}}}^* \: \dd \mu = \delta_{n=\ell} \func{\delta}{\nu = \lambda \bmod{1}}. \label{eq:SmatUnit}
\end{gather}
Its square may also be identified with conjugation at the level of the standard characters:
\begin{equation}
\sum_{m \in \ZZ} \int_{\RR / \ZZ} \Smat{\sfmod{\ell}{\Typ{\lambda}}}{\sfmod{m}{\Typ{\mu}}} \Smat{\sfmod{m}{\Typ{\mu}}}{\sfmod{n}{\Typ{\nu}}} \: \dd \mu = \delta_{n=-\ell} \func{\delta}{\nu = -\lambda \bmod{1}}. \label{eq:SmatConj}
\end{equation}
\end{subequations}
These three familiar properties lead us to expect that substituting this integration kernel into a Verlinde formula will result in the Grothendieck fusion coefficients.

Before doing this, we need to determine the S-transformation for the atypical characters.  This follows readily from the character formulae \eqref{ch:Vac} and \thmref{thm:SMatrix}.
\begin{cor} \label{cor:SMatrixAtyp}
The simple atypical characters have S-transformations
\begin{subequations} \label{eq:SMatrixAtyp}
\begin{equation}
\fch{\sfmod{\ell}{\VacMod}}{\zeta / \tau \big\vert {-1} / \tau} = \func{A}{\zeta \big\vert \tau} \sum_{m \in \ZZ} \int_{\RR / \ZZ} \Smat{\sfmod{\ell}{\VacMod}}{\sfmod{m}{\Typ{\mu}}} \fch{\sfmod{m}{\Typ{\mu}}}{\zeta \big\vert \tau} \: \dd \mu,
\end{equation}
where
\begin{equation} \label{eq:AtypSKernel}
\Smat{\sfmod{\ell}{\VacMod}}{\sfmod{m}{\Typ{\mu}}} = \brac{-1}^{\ell+m+1} \frac{\ee^{-2 \pi \ii \brac{\ell + 1/2} \mu}}{\ee^{\ii \pi \mu} - \ee^{-\ii \pi \mu}}.
\end{equation}
\end{subequations}
\end{cor}
\noindent Here, the denominator should also be regarded as shorthand for a formal power series in $\ee^{2 \pi \ii \mu}$.  In fact, it arises from summing a geometric series at its radius of convergence, a fact which may be useful to remember for the Verlinde computations to come.  We remark that both the character formulae of \eqref{ch:Vac} conveniently yield the same atypical S-transformation kernel when expressed using denominators (though the respective convergence regions are disjoint).

Finally, we address the automorphy factor $\func{A}{\zeta \big\vert \tau}$ appearing in the transformation rules \eqref{eq:SMatrix} and \eqref{eq:SMatrixAtyp}.  This factor does not depend upon the labels characterising the modules in the S-transformation kernel and, as with a similar (but less complicated) factor appearing in the S-transformation of integrable Kac-Moody module characters \cite{KacInf84}, it may be absorbed by augmenting the definition of characters by another variable $y$ which tracks the eigenvalue of the Cartan element $\wun$.  This eigenvalue is always $1$, so we end up multiplying all $\Ghost$-module characters by $y = \ee^{2 \pi \ii \theta}$.
\begin{prop} \label{prop:ModAct}
The transformations
\begin{equation} \label{eq:Crazy}
\modS \colon \left( \theta \middle\vert \zeta \middle\vert \tau \right) \longmapsto \left( \theta + \frac{\zeta^2}{2 \tau} - \frac{\zeta}{2 \tau} + \frac{\zeta}{2} + \frac{1}{2 \pi} \brac{\arg \tau - \frac{\pi}{2}} \middle\vert \frac{\zeta}{\tau} \middle\vert -\frac{1}{\tau} \right), \quad 
\modT \colon \left( \theta \middle\vert \zeta \middle\vert \tau \right) \longmapsto \left( \theta + \frac{1}{12} \middle\vert \zeta \middle\vert \tau + 1 \right) 
\end{equation}
define an action of the modular group $\SLG{SL}{2;\ZZ}$.  That is, $\modS^2 = \brac{\modS \modT}^3 = \modC$ and $\modC^2$ is the identity.
\end{prop}
\noindent The proof is a straightforward verification that $\modS^2$ and $\brac{\modS \modT}^3$ map $\left( \theta \middle\vert \zeta \middle\vert \tau \right)$ to $\left( \theta + \zeta \middle\vert -\zeta \middle\vert \tau \right)$; this obviously squares to the identity.  We remark that the term involving $\arg \tau$ in \eqref{eq:Crazy} accounts for the factor of $\abs{\tau} / -\ii \tau$ in $\func{A}{\zeta \big\vert \tau}$.\footnote{This $\tau$-dependent factor was also present in the modular S-transformations of the standard characters of admissible level $\AKMA{sl}{2}$ \cite{CreMod12,CreMod13}, but was argued to be inconsequential as phases cancel when considering modular invariants and Verlinde computations.  A more satisfactory explanation is to absorb it into the automorphy factor $\func{A}{\zeta \big\vert \tau}$ as we have done here for the standard ghost characters.}  It now follows that inserting $y$ into characters and transforming as in \eqref{eq:Crazy} will cancel the factor $\func{A}{\zeta \big\vert \tau}$ in \eqref{eq:SMatrix} and \eqref{eq:SMatrixAtyp}.  This justifies our separation of this automorphy factor from the S-transformation kernel.

\section{The Verlinde Formula} \label{sec:Verlinde}

We define a product $\Grfuse$ on the (appropriate topological completion of the) span of the standard characters by
\begin{subequations} \label{eq:Verlinde}
\begin{equation} \label{eq:Verlinde1}
\ch{\mathcal{M}} \Grfuse \ch{\mathcal{N}} = \sum_{n \in \ZZ} \int_{\RR / \ZZ} \fuscoeff{\mathcal{M}}{\mathcal{N}}{\sfmod{n}{\Typ{\nu}}} \ch{\sfmod{n}{\Typ{\nu}}} \: \dd \nu,
\end{equation}
where the coefficients appearing in the integrand are determined by the following variant of the Verlinde formula:
\begin{equation} \label{eq:Verlinde2}
\fuscoeff{\mathcal{M}}{\mathcal{N}}{\sfmod{n}{\Typ{\nu}}} = \sum_{r \in \ZZ} \int_{\RR / \ZZ} \frac{\Smat{\mathcal{M}}{\sfmod{r}{\Typ{\rho}}} \Smat{\mathcal{N}}{\sfmod{r}{\Typ{\rho}}} \Smat{\sfmod{n}{\Typ{\nu}}}{\sfmod{r}{\Typ{\rho}}}^*}{\Smat{\VacMod}{\sfmod{r}{\Typ{\rho}}}} \: \dd \rho.
\end{equation}
\end{subequations}
We will demonstrate shortly that this product, which we call the \emph{Verlinde product}, is indeed well-defined --- \eqref{eq:Verlinde1} always gives a finite linear combination of standard characters or infinite alternating sums, the latter being interpreted as atypical simple characters.  For now, we note that the Verlinde product is commutative and associative.  The unitarity \eqref{eq:SmatUnit} of the S-transformation implies that the unit is the vacuum character $\ch{\VacMod}$.

\begin{lem} \label{lem:GrFusionSpecFlow}
The Verlinde product satisfies
\begin{equation}
\ch{\sfmod{\ell}{\mathcal{M}}} \Grfuse \ch{\sfmod{m}{\mathcal{N}}} = \sfmod{\ell+m}{\ch{\mathcal{M}} \Grfuse \ch{\mathcal{N}}},
\end{equation}
where the \rhs{} is to be interpreted as evaluating the Verlinde product in the standard basis and applying spectral flow to each basis element uniformly.
\end{lem}
\begin{proof}
This follows by noting that the S-transformation kernels \eqref{eq:TypSKernel} and \eqref{eq:AtypSKernel} may be factored as
\begin{equation}
\Smat{\sfmod{\ell}{\mathcal{M}}}{\sfmod{r}{\Typ{\rho}}} = \brac{-1}^{\ell} \ee^{-2 \pi \ii \ell \rho} \: \Smat{\mathcal{M}}{\sfmod{r}{\Typ{\rho}}},
\end{equation}
where $\mathcal{M}$ is either $\Typ{\lambda}$ or $\VacMod$.  Applying this factorisation to the kernels for $\mathcal{M}$ and $\mathcal{N}$ appearing in \eqref{eq:Verlinde2}, and then absorbing both phases into the kernel for $\sfmod{n}{\Typ{\nu}}$, we arrive at
\begin{equation}
\fuscoeff{\sfmod{\ell}{\mathcal{M}}}{\sfmod{m}{\mathcal{N}}}{\sfmod{n}{\Typ{\nu}}} = \fuscoeff{\mathcal{M}}{\mathcal{N}}{\sfmod{-\ell-m+n}{\Typ{\nu}}}.
\end{equation}
Replacing $n$ by $\ell+m+n$ now gives the desired result.
\end{proof}
\begin{thm} \label{thm:GrFusion}
The Verlinde product rules take the form
\begin{subequations} \label{GrFR}
\begin{align}
\ch{\sfmod{\ell}{\VacMod}} \Grfuse \ch{\sfmod{m}{\VacMod}} &= \ch{\sfmod{\ell+m}{\VacMod}}, \label{GrFR:AA} \\
\ch{\sfmod{\ell}{\VacMod}} \Grfuse \ch{\sfmod{m}{\Typ{\mu}}} &= \ch{\sfmod{\ell+m}{\Typ{\mu}}}, \label{GrFR:AT} \\
\ch{\sfmod{\ell}{\Typ{\lambda}}} \Grfuse \ch{\sfmod{m}{\Typ{\mu}}} &= \ch{\sfmod{\ell+m}{\Typ{\lambda + \mu}}} + \ch{\sfmod{\ell+m-1}{\Typ{\lambda + \mu}}}. \label{GrFR:TT}
\end{align}
\end{subequations}
In particular, the Verlinde multiplicities \eqref{eq:Verlinde2} are non-negative integer multiples of delta functions.
\end{thm}
\begin{proof}
By \lemref{lem:GrFusionSpecFlow}, we may assume that $\ell = m = 0$.  Then, \eqref{GrFR:AA} and \eqref{GrFR:AT} follow from the vacuum character $\ch{\VacMod}$ being the unit of the Verlinde product.  We therefore turn to the rule \eqref{GrFR:TT} and compute the coefficient
\begin{align}
\fuscoeff{\Typ{\lambda}}{\Typ{\mu}}{\sfmod{n}{\Typ{\nu}}} &= \brac{-1}^{n+1} \sum_{r \in \ZZ} \ee^{-2 \pi \ii \brac{\lambda + \mu - \nu} r} \int_{\RR / \ZZ} \brac{\ee^{2 \pi \ii \brac{n+1} \rho} - \ee^{2 \pi \ii n \rho}} \: \dd \rho \notag \\
&= \brac{\delta_{n=0} + \delta_{n=-1}} \func{\delta}{\nu = \lambda + \mu \bmod{1}}.
\end{align}
The result now follows by substituting into \eqref{eq:Verlinde1}.
\end{proof}

Because the multiplicities appearing in the Verlinde product rules are non-negative integers, the product $\Grfuse$ endows the (completion of the) $\ZZ$-span of the standard characters with a ring structure.  We call this ring the Verlinde ring.  The following assumption and conjecture are now very plausible:
\begin{conj} \label{conj:rigid}
Let $\fuse$ denote the fusion product on the $\ZZ$-span of the indecomposable $\Ghost$-modules (where addition is direct sum).  We assume that fusing with any given $\Ghost$-module defines an exact functor from this fusion ring to itself, hence that the fusion product descends to a well-defined product $\Grfuse$ on the Grothendieck group:
\begin{equation}
\Gr{\mathcal{M}} \Grfuse \Gr{\mathcal{N}} = \Gr{\mathcal{M} \fuse \mathcal{N}}.
\end{equation}
We conjecture that the product on the resulting Grothendieck ring may be identified with the Verlinde product under the group isomorphism $\Gr{\mathcal{M}} \mapsto \ch{\mathcal{M}}$.  In other words, we conjecture that this constitutes an isomorphism between the Verlinde and Grothendieck fusion rings.
\end{conj}
\noindent This conjecture holds for rational \cfts{} \cite{HuaVer05}.  We will assume from now on that this conjecture holds for the $c=2$ bosonic ghost system, so we will use Verlinde and Grothendieck fusion ring terminology interchangeably.  This amounts to supposing that the Verlinde formula \eqref{eq:Verlinde} computes the character of the fusion product:
\begin{equation} \label{eq:DefCharFusion}
\ch{\mathcal{M} \fuse \mathcal{N}} = \ch{\mathcal{M}} \Grfuse \ch{\mathcal{N}}.
\end{equation}

Of course, if the \rhs{} of a Grothendieck fusion rule is the character of a simple module, it may be lifted to a genuine fusion rule.  More generally, if a Grothendieck product is a sum of characters of modules among which no non-trivial extensions are possible, then we may again lift the result to a genuine fusion rule.  In the latter case, consideration of charges and conformal weights modulo $1$ is often sufficient to rule out indecomposable extensions.  Such considerations lead us to the following fusion rules:
\begin{cor} \label{cor:Fusion}
Assuming \conjref{conj:rigid}, the Verlinde product rules of \thmref{thm:GrFusion} imply the following fusion rules:
\begin{subequations} \label{FR}
\begin{align}
\sfmod{\ell}{\VacMod} \fuse \sfmod{m}{\VacMod} &= \sfmod{\ell+m}{\VacMod}, \label{FR:AA} \\
\sfmod{\ell}{\VacMod} \fuse \sfmod{m}{\Typ{\mu}} &= \sfmod{\ell+m}{\Typ{\mu}} & &\text{(\(\mu \notin \ZZ\)),} \label{FR:AT} \\
\sfmod{\ell}{\Typ{\lambda}} \fuse \sfmod{m}{\Typ{\mu}} &= \sfmod{\ell+m}{\Typ{\lambda + \mu}} \oplus \sfmod{\ell+m-1}{\Typ{\lambda + \mu}} & &\text{(\(\lambda + \mu \notin \ZZ\)).} \label{FR:TT}
\end{align}
\end{subequations}
\end{cor}
\noindent We remark that fusing the module $\sfmod{\ell}{\VacMod}$ with its conjugate $\conjmod{\sfmod{\ell}{\VacMod}} = \sfmod{-\ell}{\conjmod{\VacMod}} = \sfmod{-\ell-1}{\VacMod}$ does not give back the vacuum, but rather its conjugate $\conjmod{\VacMod} = \sfmod{-1}{\VacMod}$.  This is consistent with the one-point function of the identity field vanishing and that of its conjugate $\func{\omega}{z}$ being non-vanishing.

\section{Fusion} \label{sec:Fusion}

In this section, we compute the remaining fusion product involving simple modules, that of the typicals $\Typ{\lambda}$ and $\Typ{-\lambda}$ (so $[\lambda] \neq [0]$).  \thmref{thm:GrFusion} and \eqnref{eq:DefCharFusion} give the character of this fusion product if we assume (and we do) that \conjref{conj:rigid} holds:
\begin{equation} \label{GrFR:TT'}
\ch{\Typ{\lambda} \fuse \Typ{-\lambda}} = \ch{\sfmod{-2}{\VacMod}} + 2 \: \ch{\sfmod{-1}{\VacMod}} + \ch{\VacMod}.
\end{equation}
We illustrate the (convex hull of the) weights of the composition factors of this fusion product in \figref{fig:Stag} (left); here, the charge increases horizontally from right to left and the conformal weight increases from top to bottom.  To deduce the module structure, we turn to the \NGK{} fusion algorithm \cite{NahQua94,GabInd96}.  This constructs (an algebraic completion of) the fusion product of two modules as a quotient of their tensor product (over $\CC$) \cite{GabFus94b}, the action on the product being characterised by the following master equations:
\begin{subequations} \label{eq:Master}
\begin{align}
\coprod{\beta_n} &= \sum_{m=0}^n \binom{n}{m} \beta_m \otimes \wun + \wun \otimes \beta_n & &\text{(\(n \ge 0\)),} \label{eq:MB1} \\
\coprod{\beta_{-n}} &= \sum_{m=0}^{\infty} \binom{m+n-1}{n-1} \brac{-1}^m \beta_m \otimes \wun + \wun \otimes \beta_n & &\text{(\(n \ge 1\)),} \label{eq:MB2} \\
\beta_{-n} \otimes \wun &= \sum_{m=n}^{\infty} \binom{m-1}{n-1} \coprod{\beta_{-m}} + \brac{-1}^{n-1} \sum_{m=0}^{\infty} \binom{m+n-1}{n-1} \wun \otimes \beta_m & &\text{(\(n \ge 1\)),} \label{eq:MB3} \\
\coprod{\gamma_n} &= \sum_{m=1}^n \binom{n-1}{m-1} \gamma_m \otimes \wun + \wun \otimes \gamma_n & &\text{(\(n \ge 1\)),} \label{eq:MG1} \\
\coprod{\gamma_{-n}} &= \sum_{m=1}^{\infty} \binom{m+n-1}{n} \brac{-1}^{m-1} \gamma_m \otimes \wun + \wun \otimes \gamma_n & &\text{(\(n \ge 0\)),} \label{eq:MG2} \\
\gamma_{-n} \otimes \wun &= \sum_{m=n}^{\infty} \binom{m}{n} \coprod{\gamma_{-m}} + \brac{-1}^n \sum_{m=1}^{\infty} \binom{m+n-1}{n} \wun \otimes \gamma_m & &\text{(\(n \ge 0\)).} \label{eq:MG3}
\end{align}
\end{subequations}
We remark that imposing \eqref{eq:MB3} and \eqref{eq:MG3} as identities that act upon the tensor product of two modules amounts to working in the quotient of the tensor product that realises the fusion product.  

\begin{figure}
\begin{tikzpicture}[vec/.style={circle,draw=black,fill=black,inner sep=2pt,minimum size=5pt}]
\path (0,0) node[vec] (xm) {} + (0.2,0.3) node {$x^-$};
\path (-2,1) node[vec] (xp) {} + (0.3,0.2) node {$x^+$};
\path (-1,0.1) node[vec] (y) {} + (0.1,0.3) node {$y$};
\path (-1,-0.1) node[vec] (w) {} + (-0.1,-0.3) node {$w$};
\draw[] (3,0) -- (xm) -- (-3,-3)
        (-3,2) -- (xp) -- (0,-3)
        (-4,0.1) -- (y) -- (2,-2.85)
        (-4,-0.1) -- (w) -- (2,-3.1);
\end{tikzpicture}
\hspace{0.05\textwidth}
\begin{tikzpicture}[auto,>=latex,vec/.style={circle,draw=black,fill=black,inner sep=2pt,minimum size=5pt}]
\path (0,0) node[vec] (xm) {}
      (-2,1) node[vec] (xp) {}
      (-1,0.1) node[vec] (y) {}
      (-1,-0.1) node[vec] (w) {};
\draw[gray] (3,0) -- (xm) -- (-3,-3)
        (-3,2) -- (xp) -- (0,-3)
        (-4,0.1) -- (y) -- (2,-2.85)
        (-4,-0.1) -- (w) -- (2,-3.1);
\draw[->,thick] (y) to[out=105,in=-15] node[swap] {$\beta_1$} (xp);
\draw[->,thick] (xp) to[out=-75,in=165] node[swap] {$\gamma_{-1}$} (w);
\draw[->,thick] (y) to[out=15,in=150] node[] {$\gamma_0$} (xm);
\draw[->,thick] (xm) to[out=-150,in=-15] node[] {$\beta_0$} (w);
\end{tikzpicture}
\caption{The structure of the fusion product $\Typ{\lambda} \fuse \Typ{-\lambda} = \sfmod{-1}{\Stag}$.  At left, the four composition factors are visualised with one vector of each marked.  The weights $(j,h)$ of the vectors $x^-$, $w$, $y$ and $x^+$ are $(0,0)$, $(1,0)$, $(1,0)$ and $(2,-1)$, respectively.  At right, the composition factors are ``glued'' together into an indecomposable module through the indicated action of the algebra modes.} \label{fig:Stag}
\end{figure}
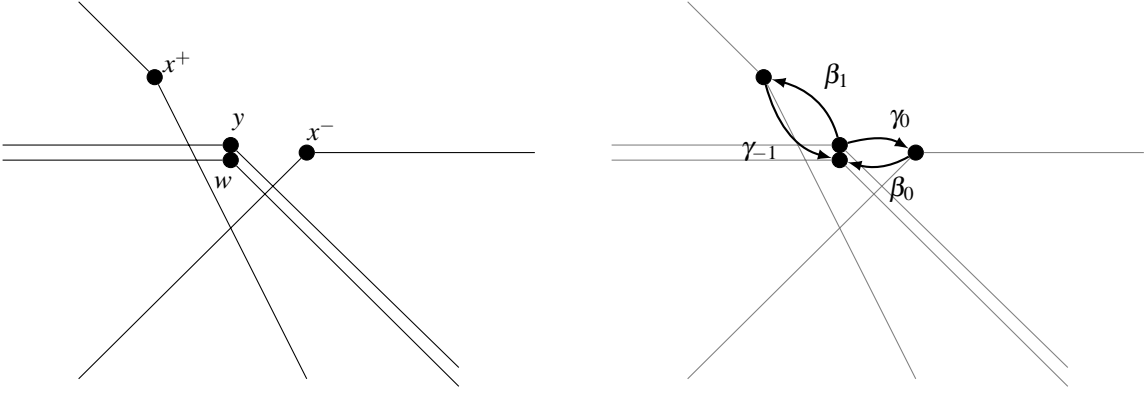

The fusion product itself will not be constructed explicitly, but we will analyse certain quotients upon which a chosen subalgebra of products of modes acts trivially.  One subalgebra that is traditionally relevant to fusion computations is that generated by the $\beta_{-m}$ and $\gamma_{-n}$, with $m \ge 1$ and $n \ge 0$; quotienting by its action defines the \emph{special subspace} \cite{NahQua94}.  Unfortunately, the typical modules $\Typ{\lambda}$ have trivial special subspaces because $\gamma_0$ acts surjectively.  The $\sfmod{\ell}{\Typ{\lambda}}$ likewise have trivial special subspaces.

The standard methodology therefore needs refining.  We introduce a (commutative) subalgebra $\alg{U}$ of the \uea{} of $\Ghost$ by
\begin{equation}
\alg{U} = \CC[\beta_{-1}, \beta_{-2}, \ldots, \gamma_{-1}, \gamma_{-2}, \ldots]
\end{equation}
and claim that
\begin{equation} \label{eq:PreFusProd}
\frac{\Typ{\lambda} \fuse \Typ{-\lambda}}{\func{\alg{U}}{\Typ{\lambda} \fuse \Typ{-\lambda}}} \subseteq \frac{\Typ{\lambda}}{\func{\alg{U}}{\Typ{\lambda}}} \otimes \frac{\Typ{-\lambda}}{\func{\alg{U}}{\Typ{-\lambda}}},
\end{equation}
as vector spaces.\footnote{Being vector spaces, one may also regard the \lhs{} of \eqref{eq:PreFusProd} as a quotient of the \rhs{}.  We present \eqref{eq:PreFusProd} as an inclusion as this is how we will prove it.  The equivalent point of view, where we instead regard the \lhs{} as a quotient, is used when actually computing a fusion product.  Then, one first characterises the \lhs{} by determining elements, called \emph{spurious states}, of the \rhs{} which must be set to $0$ for the master equations \eqref{eq:Master} to have a well-defined action.}  Because we impose \eqref{eq:MB3} and \eqref{eq:MG3} as identities on $\Typ{\lambda} \otimes \Typ{-\lambda}$, we may identify the \lhs{} with the corresponding tensor product quotient:
\begin{equation} \label{eq:FusProd}
\frac{\Typ{\lambda} \fuse \Typ{-\lambda}}{\func{\alg{U}}{\Typ{\lambda} \fuse \Typ{-\lambda}}} \cong \frac{\Typ{\lambda} \otimes \Typ{-\lambda}}{\func{\left\langle \eqref{eq:MB3}, \eqref{eq:MG3}, \coprod{\alg{U}} \right\rangle}{\Typ{\lambda} \otimes \Typ{-\lambda}}} \subseteq \frac{\Typ{\lambda}}{\func{\alg{U}}{\Typ{\lambda}}} \otimes \frac{\Typ{-\lambda}}{\func{\alg{U}}{\Typ{-\lambda}}},
\end{equation}
It is this inclusion of quotients of tensor products that we shall actually prove.

The proof amounts to showing that any $u \otimes v$, with $u \in \Typ{\lambda}$ and $v \in \Typ{-\lambda}$, representing the \lhs{} may be written as a linear combination of the $u_j \otimes v_k$ that represent the \rhs{}.  Here, the $u_j \in \Typ{\lambda}$ and $v_k \in \Typ{-\lambda}$ are the parabolic \hwvs{} that restrict to the basis vectors of the $\FinGhost$-modules $\FinTyp{\lambda}$ and $\FinTyp{-\lambda}$, respectively (see \propref{prop:FinGhReps}).  As these modules are simple, we may parametrise them so that
\begin{equation}
\begin{aligned}
\beta_0 u_j &= j u_{j+1}, & \gamma_0 u_j &= u_{j-1}, \\
\beta_0 v_k &= k v_{k+1}, & \gamma_0 v_k &= v_{k-1}
\end{aligned}
\qquad \Ra \qquad 
\begin{aligned}
J_0 u_j &= j u_j, & L_0 u_j &= 0, \\
J_0 v_k &= k v_k, & L_0 v_k &= 0
\end{aligned}
\qquad 
\begin{aligned}
\text{(\(j \in \ZZ + \lambda\)),} \\
\text{(\(k \in \ZZ - \lambda\)).}
\end{aligned}
\end{equation}
The proof proceeds in four steps, starting with some arbitrary $u \otimes v \in \Typ{\lambda} \otimes \Typ{-\lambda}$ and iterating each step on each of the terms, which we shall typically also denote by $u \otimes v$, obtained in the previous step:
\begin{enumerate}
\item If $u = \beta_{-n} u'$, with $n \ge 1$, then use \eqref{eq:MB3} to write $u \otimes v = \brac{-1}^{n-1} \sum_{m=0}^{\infty} \binom{m+n-1}{n-1} u' \otimes \beta_m v$.  Iterate this repeatedly until the result is a finite linear combination of vectors of the form $u \otimes v$, where each $u$ cannot be written as $\beta_{-n} u'$, with $n \ge 1$.  Termination is guaranteed as the conformal weight of the first factor decreases strictly with each iteration. \label{it:lb}
\item If, in any of these $u \otimes v$, we have $u = \gamma_{-n} u'$, with $n \ge 1$, then use \eqref{eq:MG3} to write each as the linear combination $\brac{-1}^n \sum_{m=1}^{\infty} \binom{m+n-1}{n} u \otimes \gamma_m v$.  Simplifying, and repeating for all terms, we arrive at a finite linear combination of vectors of the form $u_j \otimes v$. \label{it:lg}
\item If $v = \beta_{-n} v'$, with $n \ge 1$, then use \eqref{eq:MB2} and $\fcoprod{\beta_{-n}}{u_j \otimes v'} = 0$ to obtain $u_j \otimes v = -j u_{j+1} \otimes v'$.  Repeat. \label{it:rb}
\item Finally, if $v = \gamma_{-n} v'$, with $n \ge 1$, then use \eqref{eq:MG2} and $\fcoprod{\gamma_{-n}}{u_j \otimes v'} = 0$ to obtain $u_j \otimes v = 0$.  The final result is now a finite linear combination of vectors of the form $u_j \otimes v_k$, completing the proof. \label{it:rg}
\end{enumerate}
In principle, we could also apply \ref{it:lg} when $u = \gamma_0 u'$.  However, all vectors $u \in \Typ{\lambda}$ have this form, so repeating this step would lead to an infinite regress.  Instead, we apply \eqnDref{eq:MG2}{eq:MG3}, both for $n=0$, to reduce the basis $\set{u_j \otimes v_k \st j \in \ZZ + \lambda, \ k \in \ZZ - \lambda}$ of the \rhs{} of \eqref{eq:FusProd}, giving an analogue of a spurious state:
\begin{equation} \label{eq:Spur}
u_{j-1} \otimes v_{k+1} = \gamma_0 u_j \otimes v_{k+1} = \fcoprod{\gamma_0}{u_j \otimes v_{k+1}} = u_j \otimes \gamma_0 v_{k+1} = u_j \otimes v_k.
\end{equation}
We therefore propose that a basis for the \lhs{} is $\set{u_{\lambda} \otimes v_k \st k \in \ZZ - \lambda}$.

The action of $\beta_0$ and $\gamma_0$ on these basis vectors is easily computed using \eqref{eq:MB1}, \eqref{eq:MG2} and \eqref{eq:Spur}:
\begin{equation} \label{eq:Action}
\fcoprod{\beta_0}{u_{\lambda} \otimes v_k} = \lambda u_{\lambda+1} \otimes v_k + k u_{\lambda} \otimes v_{k+1} = \brac{\lambda + k} u_{\lambda} \otimes v_{k+1}, \qquad 
\fcoprod{\gamma_0}{u_{\lambda} \otimes v_k} = u_{\lambda} \otimes v_{k-1}.
\end{equation}
This is the same action as that of $\beta_0$ and $\gamma_0$ on the quotient $\Typ{0}^+ / \func{\alg{U}}{\Typ{0}^+}$ (which coincides with that on the $\FinGhost$-module $\FinTyp{0}^+$ appearing in \propref{prop:FinGhReps}):  $\gamma_0$ acts surjectively while $\beta_0$ annihilates the vector $u_{-k} \otimes v_k$ of weight $(0,0)$.  We therefore conclude that the fusion product $\Typ{\lambda} \fuse \Typ{-\lambda}$ has a quotient isomorphic to $\Typ{0}^+$.  This accounts for the composition factor $\VacMod$ and one of the $\sfmod{-1}{\VacMod}$ factors appearing in \eqref{GrFR:TT'}.  It also verifies the arrow labelled by $\gamma_0$ in \figref{fig:Stag} (right).

Because $\gamma_{-1} \in \alg{U}$ acts surjectively on the composition factor $\sfmod{-2}{\VacMod}$, every vector associated to this factor is set to $0$ in the fusion quotient that we have computed.  It therefore remains to account for the other composition factor $\sfmod{-1}{\VacMod}$ in \eqref{GrFR:TT'}.  As no vectors associated to this factor are observed in the fusion quotient, they must be in the image of $\alg{U}$.  In particular, the vector of weight $(1,0)$ that is labelled by $w$ in \figref{fig:Stag} (left) must be in $\im \alg{U}$.  Referring to this figure (or considering multiplicities from the character \eqref{GrFR:TT'} of the fusion product), we see that the only way this can happen is if $w$ is a non-zero multiple of $\gamma_{-1} x^+$.  This conclusion therefore verifies the arrow labelled by $\gamma_{-1}$ drawn in \figref{fig:Stag} (right).

We remark that if the basis proposed after \eqref{eq:Spur} were incorrect, meaning that there were further spurious states to find, then we would have to set some of the elements of $\Typ{0}^+ / \func{\alg{U}}{\Typ{0}^+}$ to $0$.  However, this is impossible because \figref{fig:Stag} makes it clear that there cannot be any $\alg{U}$-descendants beyond those we have accounted for.  The basis is therefore correct.

To obtain the remaining arrows in \figref{fig:Stag} (right), we change the subalgebra by whose action we quotient.  Let $\alg{U}'$ denote the (commutative) subalgebra
\begin{equation}
\alg{U}' = \CC[\beta_0, \beta_{-1}, \beta_{-2}, \ldots, \gamma_{-2}, \gamma_{-3}, \gamma_{-4}, \ldots].
\end{equation}
The claim is now that
\begin{equation} \label{eq:FusProd'}
\frac{\Typ{\lambda} \fuse \Typ{-\lambda}}{\func{\alg{U}'}{\Typ{\lambda} \fuse \Typ{-\lambda}}} \subseteq \frac{\Typ{\lambda}}{\func{\alg{U}}{\Typ{\lambda}}} \otimes \frac{\Typ{-\lambda}}{\func{\alg{U}'}{\Typ{-\lambda}}};
\end{equation}
that is, that any $u \otimes v \in \Typ{\lambda} \otimes \Typ{-\lambda}$ may be reduced to a linear combination of vectors of the form $u_j \otimes \gamma_{-1}^m v_k$.  The proof again proceeds as above, with the same proviso regarding \eqnDref{eq:MB3}{eq:MG3}, though \ref{it:lg} and \ref{it:rg} are now only performed when $n \ge 2$.  Moreover, we need an additional step after \ref{it:lg}:
\begin{enumerate}[label=(\arabic*'),start=2]
\item As $u = \gamma_{-1}^{\ell} u_j$, we use \eqref{eq:MG2} and \eqref{eq:MG3} to write $\gamma_{-1}^{\ell} u_j \otimes v = \gamma_{-1}^{\ell-1} u_j \otimes \gamma_{-1} v - \sum_{m=1}^{\infty} m \gamma_{-1}^{\ell-1} u_j \otimes \gamma_m v$, when $\ell > 0$.  Repeat until we have a finite linear combination of vectors of the form $u_j \otimes v$.
\end{enumerate}
We may again reduce the basis for the \rhs{} of \eqref{eq:FusProd'} by computing analogues of spurious states:
\begin{subequations}
\begin{gather}
0 = \fcoprod{\beta_0}{v_j \otimes \gamma_{-1}^m w_k} = j v_{j+1} \otimes \gamma_{-1}^m w_k + k v_j \otimes \gamma_{-1}^m w_{k+1}, \label{eq:Spur'} \\
v_j \otimes \gamma_{-1}^m w_k = \gamma_0 v_{j+1} \otimes \gamma_{-1}^m w_k = \func{\brac{\coprod{\gamma_0} + \coprod{\gamma_{-1}}}}{v_{j+1} \otimes \gamma_{-1}^m w_k} = v_{j+1} \otimes \gamma_{-1}^m w_{k-1} + v_{j+1} \otimes \gamma_{-1}^{m+1} w_k. \label{eq:Spur''}
\end{gather}
\end{subequations}
Applying \eqref{eq:Spur''} repeatedly lets us reduce the power of $\gamma_{-1}$ to $0$, then \eqref{eq:Spur'} lets us fix $j = \lambda$.  Our proposed basis is therefore $\set{u_{\lambda} \otimes v_k \st k \in \ZZ - \lambda}$.  We now compute
\begin{subequations}
\begin{align}
\fcoprod{\beta_1}{v_{\lambda} \otimes w_k} &= \lambda v_{\lambda+1} \otimes w_k = -k v_{\lambda} \otimes w_{k+1}, \\
\fcoprod{\gamma_{-1}}{v_{\lambda} \otimes w_k} &= v_{\lambda} \otimes \gamma_{-1} w_k = v_{\lambda-1} \otimes w_k - v_{\lambda} \otimes w_{k-1} = -\frac{\lambda+k-2}{k-1} v_{\lambda} \otimes w_{k-1}
\end{align}
\end{subequations}
and a little work shows that this action matches that on the quotient $\sfmod{-1}{\Typ{0}^-} / \func{\alg{U}'}{\sfmod{-1}{\Typ{0}^-}}$.  This verifies the arrow labelled by $\beta_1$ in \figref{fig:Stag} (right) and that labelled by $\beta_0$ is obtained by noting that the missing vector of weight $(1,0)$ can only be a $\alg{U}'$-descendant of the vector of weight $(0,0)$.

It remains only to determine if there are any ambiguities in the structure that we have uncovered for this fusion product $\Typ{\lambda} \fuse \Typ{-\lambda}$.  The analysis amounts to considering the four vectors labelled in \figref{fig:Stag} (left):
\begin{itemize}
\item First, choose $x^+ \neq 0$ of weight $(2,-1)$.
\item Then, define $w = \gamma_{-1} x^+$ so that $w$ has weight $(1,0)$.  Let $y$ and $w$ be linearly independent in this weight space.
\item Fix $x^-$, of weight $(0,0)$, by requiring that $\beta_0 x^- = w$.
\end{itemize}
We will fix the normalisation of $y$ shortly.  For now, we note that
\begin{equation}
J_0 y = \brac{\gamma_0 \beta_0 + \gamma_{-1} \beta_1} y = y + \brac{\beta_0 \gamma_0 + \gamma_{-1} \beta_1} y, \qquad 
L_0 y = -\gamma_{-1} \beta_1 y \neq 0,
\end{equation}
so that $\brac{J_0 - \wun} y = \brac{\beta_0 \gamma_0 + \gamma_{-1} \beta_1} y$ and $L_0 y$ are proportional to $w$ (see \figref{fig:Stag}).  The Virasoro zero mode therefore has a Jordan block of rank $2$ indicating that $\Typ{\lambda} \fuse \Typ{-\lambda}$ is a staggered module in the sense of \cite{RidSta09,CreLog13}.  We may now normalise $y$ so that
\begin{itemize}
\item $L_0 y = w$,
\end{itemize}
noting that this fixes $y$ up to adding multiples of $w$.  The structure of the staggered module is then determined by computing $\beta_1 y = b_+ x^+$ and $\gamma_0 y = b_- x^-$, as the constants $b_{\pm}$ are independent of the remaining freedom in choosing $y$.  We find that
\begin{equation}
w = L_0 y = -\gamma_{-1} \beta_1 y = -b_+ \beta_{-1} x^+ = -b_+ w \qquad \Ra \qquad b_+ = -1.
\end{equation}
To compute $b_-$, we note that the coproduct formula $\coprod{J_0} = J_0 \otimes \wun + \wun \otimes J_0$ implies that $J_0$ acts semisimply on the fusion product $\Typ{\lambda} \fuse \Typ{-\lambda}$ because it does on the typical modules.  Thus, we deduce that
\begin{equation}
0 = \brac{J_0 - \wun} y = \brac{\beta_0 \gamma_0 + \gamma_{-1} \beta_1} y = \brac{b_- + b_+} w \qquad \Ra \qquad b_- = -b_+ = 1.
\end{equation}
The analysis of the fusion product is complete and we summarise the result as follows:
\begin{thm} \label{thm:FusionTT}
The ghost fusion rule
\begin{equation}
\Typ{\lambda} \fuse \Typ{-\lambda} = \sfmod{-1}{\Stag}
\end{equation}
defines an indecomposable staggered module $\Stag$ with rank $2$ Jordan blocks that is determined up to isomorphism by either of the following exact sequences or by its Loewy diagram:
\begin{equation}
\begin{aligned}
&\dses{\sfmod{}{\Typ{0}^-}}{}{\Stag}{}{\Typ{0}^-}, \\
&\dses{\Typ{0}^+}{}{\Stag}{}{\sfmod{}{\Typ{0}^+}},
\end{aligned}
\mspace{100mu}
\parbox{0.3\textwidth}{
\begin{tikzpicture}[thick,>=latex,
                    nom/.style={circle,draw=black!20,fill=black!20,inner sep=1pt}
                    ]
\node (top) at (0,1.5) [] {\(\VacMod\)};
\node (left) at (-1.5,0) [] {\(\sfmod{-1}{\VacMod}\)};
\node (right) at (1.5,0) [] {\(\sfmod{}{\VacMod}\).};
\node (bot) at (0,-1.5) [] {\(\VacMod\)};
\node at (0,0) [nom] {\(\Stag\)};
\draw[->] (top) -- (left);
\draw[->] (top) -- (right);
\draw[->] (left) -- (bot);
\draw[->] (right) -- (bot);
\end{tikzpicture}
}
\end{equation}
\end{thm}
\noindent In other words, there are no logarithmic couplings \cite{RidPer07} to determine in order to completely specify the isomorphism class of $\Stag$.  We emphasise that this computation assumed \conjref{conj:rigid}.

It is extremely natural to generalise this result to the fusion rules of spectrally-flowed typical modules.  This requires the following standard conjecture, still unproven to the best of our knowledge, that lifts \lemref{lem:GrFusionSpecFlow} to fusion:
\begin{conj} \label{conj:fuse}
The fusion product satisfies
\begin{equation} \label{eq:fuse}
\sfmod{\ell}{\mathcal{M}} \fuse \sfmod{m}{\mathcal{N}} \cong \sfmod{\ell + m}{\mathcal{M} \fuse \mathcal{N}}.
\end{equation}
\end{conj}
\begin{cor} \label{cor:FusionTT}
Assuming \conjDref{conj:rigid}{conj:fuse}, \thmref{thm:FusionTT} implies the following fusion rules:
\begin{equation}
\sfmod{\ell}{\Typ{\lambda}} \fuse \sfmod{m}{\Typ{-\lambda}} = \sfmod{\ell + m-1}{\Stag}.
\end{equation}
\end{cor}
\noindent Because the spectrum of the theory contains staggered modules, the bosonic ghost system at $c=2$ is a \lcft{}.  We remark that fusing a typical module with its conjugate does not give the conjugate to the vacuum module, but rather a staggered module that covers the conjugate vacuum module.  The fusion rules involving the staggered modules now follow from associativity.
\begin{cor} \label{cor:FusionS}
We have the following fusion rules, assuming \conjDref{conj:rigid}{conj:fuse}:
\begin{subequations}
\begin{align}
\sfmod{\ell}{\VacMod} \fuse \sfmod{m}{\Stag} &= \sfmod{\ell+m}{\Stag}, \\
\sfmod{\ell}{\Typ{\lambda}} \fuse \sfmod{m}{\Stag} &= \sfmod{\ell+m+1}{\Typ{\lambda}} \oplus 2 \: \sfmod{\ell+m}{\Typ{\lambda}} \oplus \sfmod{\ell+m-1}{\Typ{\lambda}}, \\
\sfmod{\ell}{\Stag} \fuse \sfmod{m}{\Stag} &= \sfmod{\ell+m+1}{\Stag} \oplus 2 \: \sfmod{\ell+m}{\Stag} \oplus \sfmod{\ell+m-1}{\Stag}.
\end{align}
\end{subequations}
\end{cor}
\noindent We remark that there are many other indecomposables whose fusion rules have not been determined, the atypical standards $\Typ{0}^{\pm}$ and the length $3$ subquotients of $\Stag$, for example.  We expect that computing these fusion products iteratively will fill out a complete set of indecomposables for the $c=2$ ghost theory, much as one finds in the case of $\AKMSA{gl}{1}{1}$ \cite{GotRep07}.  As the results determined above seem to suggest that the typical modules $\sfmod{\ell}{\Typ{\lambda}}$ and staggered modules $\sfmod{\ell}{\Stag}$ form an ideal in the fusion ring, we make the following conjecture:
\begin{conj} \label{conj:Proj}
Let $\categ{C}$ be the abelian category of ghost \voa{} modules generated, by imposing closure under extensions, from the typicals $\sfmod{\ell}{\Typ{\lambda}}$ and the simple atypicals $\sfmod{\ell}{\VacMod}$ (we still insist that $\wun \in \Ghost$ act as the identity on these extensions).  Then, in $\categ{C}$, the typical module $\sfmod{\ell}{\Typ{\lambda}}$ is simple and projective, whereas the staggered module $\sfmod{\ell}{\Stag}$ is the projective cover of the simple atypical module $\sfmod{\ell}{\VacMod}$.
\end{conj}
\noindent The category $\categ{C}$ of ghost \voa{} modules is then closed under fusion and conjugation.  Moreover, we will see shortly that one can construct modular invariant partition functions from the characters of its modules.  We therefore think of $\categ{C}$ as being the physically relevant module category for bosonic ghost (logarithmic) \cfts{}.  It seems very likely to us that this category is rigid, so that, for example, fusing with any given module defines an exact functor from $\categ{C}$ to itself.  Fusion would then define a well-defined product of the Grothendieck group, proving half of \conjref{conj:rigid}.  We hope to return to this question of rigidity in the future.

Finally, we remark that in order to explicitly observe the Jordan block for $L_0$ using the \NGK{} algorithm, one would have to construct a quotient in which $w \neq 0$.  This would require excluding all powers of $\beta_0$ and $\gamma_{-1}$ from the subalgebra by whose action we quotient; the largest such subalgebra is that generated by the $\beta_n$ and $\gamma_{n-1}$ with $n \le -1$.  Unfortunately, the quotient of $\Typ{\lambda} \fuse \Typ{-\lambda}$ by the action of this subalgebra has infinite-dimensional subspaces of constant charge.  Thus, linear algebra would not suffice to determine the existence of the Jordan block, leading one instead into the world of abstract analysis.  We will also leave this technical endeavour for the future.

\section{Modular Invariants} \label{sec:ModInv}

Since the S-transformation is symmetric and unitary in the standard basis (\secref{sec:Mod}), the diagonal partition function
\begin{equation}
\func{Z_{\text{diag.}}}{y;z;q} = \sum_{\ell \in \ZZ} \int_{\RR / \ZZ} \abs{\fch{\sfmod{\ell}{\Typ{\lambda}}}{y;z;q}}^2 \: \dd \lambda
\end{equation}
is (formally) modular invariant.  Here, it is important to augment the characters by the additional variable $y$ as in the discussion surrounding \propref{prop:ModAct}.  According to the proposals of \cite{QueFre07,CreMod13}, the corresponding bulk state space should have the form
\begin{equation}
\mathbf{H} = \mathbf{B} \oplus \bigoplus_{\ell \in \ZZ} \directint_{\RR / \ZZ} \sfmod{\ell}{\Typ{\lambda}} \otimes \sfmod{\ell}{\Typ{\lambda}} \: \dd \lambda,
\end{equation}
where $\mathbf{B}$ is an indecomposable atypical bulk module whose structure is described by the following (partial) Loewy diagram in which the solid and dotted arrows represent the action of the two copies of $\Ghost$:
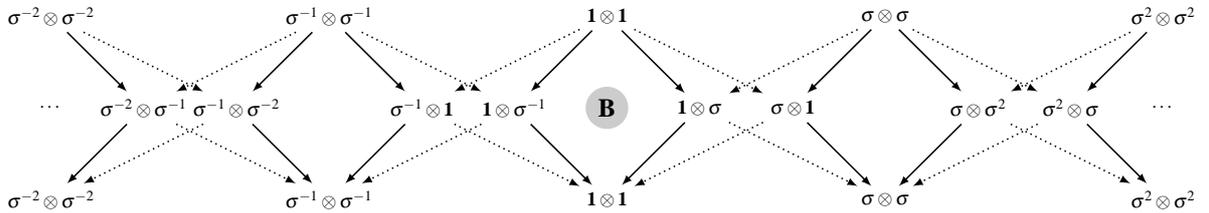
\begin{center}
\scalebox{0.73}{
\begin{tikzpicture}[thick,>=latex,scale=1.67,
                    nom/.style={circle,draw=black!20,fill=black!20,inner sep=3pt}
                    ]
\node at (-6,0) {$\cdots$};
\node at (0,0) [nom] {\scalebox{1.33}{$\mathbf{B}$}};
\node at (6,0) {$\cdots$};
\node (t1) at (-6,1) {$\sfaut^{-2} \otimes \sfaut^{-2}$};
\node (t2) at (-3,1) {$\sfaut^{-1} \otimes \sfaut^{-1}$};
\node (t3) at (0,1) {$\wun \otimes \wun$};
\node (t4) at (3,1) {$\sfaut \otimes \sfaut$};
\node (t5) at (6,1) {$\sfaut^2 \otimes \sfaut^2$};
\node (m1) at (-5,0) {$\sfaut^{-2} \otimes \sfaut^{-1}$};
\node (m2) at (-4,0) {$\sfaut^{-1} \otimes \sfaut^{-2}$};
\node (m3) at (-2,0) {$\sfaut^{-1} \otimes \wun$};
\node (m4) at (-1,0) {$\wun \otimes \sfaut^{-1}$};
\node (m5) at (1,0) {$\wun \otimes \sfaut$};
\node (m6) at (2,0) {$\sfaut \otimes \wun$};
\node (m7) at (4,0) {$\sfaut \otimes \sfaut^2$};
\node (m8) at (5,0) {$\sfaut^2 \otimes \sfaut$};
\node (b1) at (-6,-1) {$\sfaut^{-2} \otimes \sfaut^{-2}$};
\node (b2) at (-3,-1) {$\sfaut^{-1} \otimes \sfaut^{-1}$};
\node (b3) at (0,-1) {$\wun \otimes \wun$};
\node (b4) at (3,-1) {$\sfaut \otimes \sfaut$};
\node (b5) at (6,-1) {$\sfaut^2 \otimes \sfaut^2$};
\draw[->] (t1) -- (m1);
\draw[->] (t2) -- (m2);
\draw[->] (t2) -- (m3);
\draw[->] (t3) -- (m4);
\draw[->] (t3) -- (m5);
\draw[->] (t4) -- (m6);
\draw[->] (t4) -- (m7);
\draw[->] (t5) -- (m8);
\draw[->] (m1) -- (b1);
\draw[->] (m2) -- (b2);
\draw[->] (m3) -- (b2);
\draw[->] (m4) -- (b3);
\draw[->] (m5) -- (b3);
\draw[->] (m6) -- (b4);
\draw[->] (m7) -- (b4);
\draw[->] (m8) -- (b5);
\draw[->,dotted] (t1) -- (m2);
\draw[->,dotted] (t2) -- (m1);
\draw[->,dotted] (t2) -- (m4);
\draw[->,dotted] (t3) -- (m3);
\draw[->,dotted] (t3) -- (m6);
\draw[->,dotted] (t4) -- (m5);
\draw[->,dotted] (t4) -- (m8);
\draw[->,dotted] (t5) -- (m7);
\draw[->,dotted] (m2) -- (b1);
\draw[->,dotted] (m1) -- (b2);
\draw[->,dotted] (m4) -- (b2);
\draw[->,dotted] (m3) -- (b3);
\draw[->,dotted] (m6) -- (b3);
\draw[->,dotted] (m5) -- (b4);
\draw[->,dotted] (m8) -- (b4);
\draw[->,dotted] (m7) -- (b5);
\end{tikzpicture}
}
\end{center}
Here, we represent the bulk composition factor $\sfmod{\ell}{\VacMod} \otimes \sfmod{m}{\VacMod}$ by the automorphism $\sfaut^{\ell} \otimes \sfaut^m$ for brevity.  The character
\begin{equation}
\ch{\mathbf{B}} = \sum_{\ell \in \ZZ} \ch{\sfmod{\ell}{\VacMod}}^* \ch{\sfmod{\ell}{\Stag}} = \sum_{\ell \in \ZZ} \ch{\sfmod{\ell}{\Stag}}^* \ch{\sfmod{\ell}{\VacMod}}
\end{equation}
underscores the similarity between this proposal and the standard decompositions of the regular representations of finite-dimensional associative algebras and compact Lie groups.  We note that the nilpotent part of the actions of $L_0$ and $\ahol{L}_0$ both map each vector associated with the head of this module (the top composition factors) to the same vector in its socle (the bottom composition factors).  Locality, meaning the single-valuedness of bulk correlators, is thus satisfied for this proposed bulk module structure \cite{GabLoc99}.

Note that the charge conjugate partition function is likewise formally modular invariant, but the corresponding atypical bulk module does not have a submodule isomorphic to $\VacMod \otimes \VacMod$ because $\VacMod \ncong \conjmod{\VacMod}$.  In particular, the charge conjugate bulk state space would possess no vacuum state, so its physical consistency is not clear to us.  It would be interesting to know whether these bulk state space proposals may be interpreted in terms of coends as advocated in \cite{FucFro13}.

There are nevertheless many other modular invariants of simple current type.  Indeed, the fusion rules \eqref{FR:AA} show that each of the $\sfmod{p}{\VacMod}$ is a simple current of infinite order.  The vacuum module $\ExtVacMod$ of the corresponding simple current extension $\ExtGhost{p}$ (we take $p>0$ without loss of generality) then decomposes as
\begin{equation}
\ExtVacMod \cong \bigoplus_{r \in \ZZ} \sfmod{rp}{\VacMod},
\end{equation}
when restricted to a $\Ghost$-module.  It is easy to check that the charges and conformal weights of the vectors in $\ExtVacMod$ are integers and that this continues to hold for all atypical indecomposables.  The same is not true for the typical extended algebra modules.  The module
\begin{equation}
\ExtTyp{\lambda} \cong \bigoplus_{r \in \ZZ} \sfmod{rp}{\Typ{\lambda}}
\end{equation}
turns out to be untwisted, meaning that the extended algebra fields have trivial monodromy, if and only if $p \lambda \in \ZZ$.

When $p$ is even, the characters of the (untwisted) standard extended algebra modules $\sfmod{\ell}{\ExtTyp{j/p}}$, for $j, \ell = 0, 1, \ldots, p-1$, span a finite-dimensional representation of the modular group.\footnote{When $p$ is odd, these untwisted characters are transformed by S into a linear combination of twisted characters.  We expect that in this case the extended algebra $\ExtGhost{p}$ is fermionic in nature.}  In particular,
\begin{equation}
\ch{\sfmod{\ell}{\ExtTyp{j/p}}} \overset{\modS}{\longmapsto} \frac{1}{p} \sum_{m,k=0}^{p-1} \brac{-1}^{\ell+m} \ee^{-2 \pi \ii \brac{\ell k + m j} / p} \ch{\sfmod{m}{\ExtTyp{k/p}}}
\end{equation}
and the extended S-matrices are easily checked to be symmetric and unitary.  The partition function
\begin{equation}
Z_p = \sum_{\ell,j=0}^{p-1} \abs{\ch{\sfmod{\ell}{\ExtTyp{j/p}}}}^2 = \sum_{\ell,r \in \ZZ} \sum_{j=0}^{p-1} \ch{\sfmod{\ell}{\Typ{j/p}}}^* \ch{\sfmod{\ell+rp}{\Typ{j/p}}}
\end{equation}
is therefore modular invariant for $p$ even.  We remark that the corresponding theories are always logarithmic.

Finally, we mention that although these simple current extensions define formal modular invariants, their modular properties are unsatisfactory in general.  In particular, it is not clear how to evaluate the S-transforms of the simple atypical characters --- the obvious manipulations lead to a divergence due to the pole in \eqnref{eq:SMatrixAtyp}.  It would be interesting to understand this because the standard examples of $C_2$-cofinite logarithmic theories, whose modular properties are similarly unsatisfactory, may likewise be realised as simple current extensions \cite{RidMod13}.  It would also be interesting to classify all ghost modular invariants; we hope to return to these questions in the future.

\section*{Acknowledgements}

We thank J\"{u}rgen Fuchs and Christoph Schweigert for illuminating discussions regarding parabolic Verma modules and the organisers of the Erwin Schr\"{o}dinger Institute programme ``Modern trends in topological quantum field theory'' for their hospitality.
DR's research is supported by the Australian Research Council Discovery Project DP1093910.  
SW's work is supported by the Australian Research Council Discovery Early Career Researcher Award DE140101825.

\appendix

\section{Fusion for the $c=-1$ Bosonic Ghost System} \label{app:c=-1Fusion}

In this appendix, we quickly recall the fusion rules of the bosonic ghosts with $a=0$, hence central charge $c=-1$.  These were partially reported in \cite[Sec.~5]{RidFus10} as consequences of the $\AKMA{sl}{2 ; \RR}_{-1/2}$ fusion rules computed there, see also \cite[Sec.2.3]{CreMod12}.  We include them here for comparison with the $c=2$ results given in \corref{cor:Fusion}, \thmref{thm:FusionTT} and \corref{cor:FusionS}.  The equivalence of the results then supports our assertion that the tensor structure on the category of ghost modules is independent of the central charge.\footnote{We note that to check this completely, we would have to compute the fusion rules of the other indecomposables, for example those of the length $3$ subquotients of the $\sfmod{\ell}{\Stag}$ and the ``zigzag''-type indecomposables (see \cite{GotRep07} for example) that we expect to result from fusing these.}

As noted in \secref{sec:Algebra}, the different ghost systems only differ in the choice of conformal structure, so their module categories are equivalent (as abelian categories).  However, there is one important difference:  Because $h_{\beta}^0 = h_{\gamma}^0 = \tfrac{1}{2}$ when $a=0$, one should also consider spectral flow twists $\sfaut^{\ell}$ where $\ell$ is a half-integer.\footnote{In particular, the spectral flow automorphism $\gamma$ used in \cite{RidFus10} must be identified with $\sfaut^{1/2}$ here.}  This translates into integer spectral flow twists for the $\ZZ_2$-orbifold $\AKMA{sl}{2 ; \RR}_{-1/2}$.  The results of \cite{RidFus10} assumed \conjref{conj:fuse} and may be stated in the following form:
\begin{subequations}
\begin{gather}
\sfmod{\ell}{\OthVac} \othfuse \sfmod{m}{\OthVac} = \sfmod{\ell + m}{\OthVac}, \quad 
\sfmod{\ell}{\OthVac} \othfuse \sfmod{m}{\OthTyp{\lambda}} = \sfmod{\ell + m}{\OthTyp{\lambda}}, \quad 
\sfmod{\ell}{\OthVac} \othfuse \sfmod{m}{\OthStag} = \sfmod{\ell + m}{\OthStag}, \\
\sfmod{\ell}{\OthTyp{\lambda}} \othfuse \sfmod{m}{\OthTyp{\mu}} = 
\begin{cases}
\sfmod{\ell + m}{\OthStag} & \text{if \(\lambda + \mu \in \ZZ\),} \\
\sfmod{\ell + m + 1/2}{\OthTyp{\lambda + \mu + 1/2}} \oplus \sfmod{\ell + m - 1/2}{\OthTyp{\lambda + \mu - 1/2}} & \text{otherwise,}
\end{cases}
\\
\sfmod{\ell}{\OthTyp{\lambda}} \othfuse \sfmod{m}{\OthStag} = \sfmod{\ell + m + 1}{\OthTyp{\lambda}} \oplus 2 \: \sfmod{\ell + m}{\OthTyp{\lambda}} \oplus \sfmod{\ell + m - 1}{\OthTyp{\lambda}}, \\
\sfmod{\ell}{\OthStag} \othfuse \sfmod{m}{\OthStag} = \sfmod{\ell + m + 1}{\OthStag} \oplus 2 \: \sfmod{\ell + m}{\OthStag} \oplus \sfmod{\ell + m - 1}{\OthStag}.
\end{gather}
\end{subequations}
Here, $\othfuse$ denotes the fusion product of the $c=-1$ theory, $\OthVac$ denotes the vacuum module, the $\OthTyp{\lambda}$ constitute a family of parabolic \hwms{} parametrised by $[\lambda] \in \RR / \ZZ$ whose elements are simple if $[\lambda] \neq \tfrac{1}{2}$, and $\OthStag$ denotes a staggered module whose Loewy diagram
\begin{equation}
\parbox{0.3\textwidth}{
\begin{tikzpicture}[thick,>=latex,
                    nom/.style={circle,draw=black!20,fill=black!20,inner sep=1pt}
                    ]
\node (top) at (0,1.5) [] {\(\OthVac\)};
\node (left) at (-1.5,0) [] {\(\sfmod{-1}{\OthVac}\)};
\node (right) at (1.5,0) [] {\(\sfmod{}{\OthVac}\)};
\node (bot) at (0,-1.5) [] {\(\OthVac\)};
\node at (0,0) [nom] {\(\OthStag\)};
\draw[->] (top) -- (left);
\draw[->] (top) -- (right);
\draw[->] (left) -- (bot);
\draw[->] (right) -- (bot);
\end{tikzpicture}
}
\end{equation}
fixes its structure up to isomorphism.

The equivalence between these results and those that we have derived for $c=2$ is given by the identifications
\begin{equation}
\OthVac \longleftrightarrow \VacMod, \qquad 
\sfmod{-1/2}{\OthTyp{\lambda + 1/2}} \longleftrightarrow \Typ{\lambda}, \qquad 
\OthStag \longleftrightarrow \Stag.
\end{equation}
The twist by $\sfaut^{-1/2}$ for the standard modules should not be surprising:  \eqnref{eq:TDef} implies that conformal weights at $c=-1$ ($a=0$) and $c=2$ ($a=-\tfrac{1}{2}$) are related by
\begin{equation}
L_0^0 = L_0^{-1/2} - \frac{1}{2} J_0.
\end{equation}
Thus, the parabolic \hwvs{} of $\Typ{\lambda}$, which all have conformal weight $0$, will no longer have constant conformal weight upon changing the conformal structure.  The shift of $\lambda$ by $\tfrac{1}{2}$ likewise accounts for the fact that the atypical point is $[\lambda] = [\tfrac{1}{2}]$ for $c=-1$, rather than $[\lambda] = [0]$.  We view this identification as providing strong evidence for the equivalence of the $c=-1$ and $c=2$ tensor categories.  In fact, we believe that this equivalence should hold much more generally because fusion, at least in the \NGK{} formalism, seems to be depend only upon the translation operator $L_{-1}$, which is independent of the choice of conformal structure because of \eqref{eq:TDef} and $\brac{\pd J}_{-1} = 0$.

\flushleft
%\bibliography{ghosts}
%\bibliographystyle{unsrt}

\end{document}